\documentclass[a4paper,english,cleveref, autoref, thm-restate]{lipics-v2021}

\pdfoutput=1 
\hideLIPIcs

\title{A unified worst case for classical simplex\\ and policy iteration pivot rules\thanks{A preliminary version~\cite{disser2023unified} of this paper appeared in \emph{Proceedings of the 34th International Symposium on Algorithms and Computation (ISAAC 2023)}.}}

\titlerunning{A unified worst case for classical simplex and policy iteration pivot rules}

\author{Yann Disser}{TU Darmstadt, Germany}{disser@mathematik.tu-darmstadt.de}{https://orcid.org/0000-0002-2085-0454}{}

\author{Nils Mosis}{TU Darmstadt, Germany}{mosis@mathematik.tu-darmstadt.de}{https://orcid.org/0000-0002-0692-0647}{}

\authorrunning{Y. Disser and N. Mosis} 

\Copyright{Yann Disser and Nils Mosis}



\nolinenumbers 

\theoremstyle{definition} 
\newtheorem{definition2}[theorem]{Definition}

\usepackage{mathtools}
\usepackage{dsfont}
\newcommand{\N}{\mathds{N}}

\newcommand{\R}{\mathds{R}}

\newcommand{\mdz}{\mathcal{D}_n}
\newcommand{\mbl}{\mathcal{B}_n}
\DeclareMathOperator{\m}{m_1}
\DeclareMathOperator{\lone}{\ell_1}
\DeclareMathOperator{\leave}{leave}
\DeclareMathOperator{\stay}{stay}
\DeclareMathOperator{\enter}{enter}
\DeclareMathOperator{\skipp}{skip}
\DeclareMathOperator{\board}{board}
\DeclareMathOperator{\travel}{travel}
\DeclareMathOperator{\Valsolo}{Val}
\DeclareMathOperator{\Val}{Val_\pi}

\newcommand{\order}{\mathcal{N}_{\mbl}}
\newcommand{\dorder}{\mathcal{N}_{\mdz}}
\newcommand{\Lx}[2]{L(#1,#2)}
\DeclareMathOperator{\B}{B}

\newcommand{\nwstr}[1]{\pi^{#1}}

\usepackage{tikz}
\usetikzlibrary{arrows.meta,positioning,calc}

\usepackage[ruled,vlined]{algorithm2e}
\newcommand{\vgap}{\vspace{.2em}}
\SetKw{KwAnd}{and}
\SetKw{KwOr}{or}
\SetKw{KwNot}{not}
\DontPrintSemicolon
\SetKw{KwReturn}{return}
\SetKwFor{For}{for}{\!:}{endfor}
\SetKwFor{While}{while}{\!:}{endw}
\SetKwFor{ForEach}{for each}{\!:}{endfor}
\SetKwIF{If}{ElseIf}{Else}{if}{\!:}{else if}{else}{endif}
\SetKwComment{tcp}{(}{)}

\begin{document}

\maketitle

\begin{abstract}
We construct a family of Markov decision processes for which the policy iteration algorithm needs an exponential number of improving switches with Dantzig's rule, with Bland's rule, and with the Largest Increase pivot rule.
This immediately translates to a family of linear programs for which the simplex algorithm needs an exponential number of pivot steps with the same three pivot rules.
Our results yield a unified construction that simultaneously reproduces well-known lower bounds for these classical pivot rules, and we are able to infer that any (deterministic or randomized) combination of them cannot avoid an exponential worst-case behavior.
Regarding the policy iteration algorithm, pivot rules typically switch multiple edges simultaneously and our lower bound for Dantzig's rule and the Largest Increase rule, which perform only single switches, seem novel.
Regarding the simplex algorithm, the individual lower bounds were previously obtained separately via deformed hypercube constructions.
In contrast to previous bounds for the simplex algorithm via Markov decision processes, our rigorous analysis is reasonably concise.

\end{abstract}

\section{Introduction}

Since the simplex algorithm for linear programming was proposed by Dantzig in 1951~\cite{dantzig1951maximization}, it has been a central question in discrete optimization whether it admits a polynomial time pivot rule.
A positive answer to this question would yield an efficient combinatorial algorithm for solving linear programs, and thus resolve an open problem on Smale's list of mathematical problems for the 21st century~\cite{smale2000mathematical}.
It would also resolve the polynomial Hirsch conjecture~\cite{dantzig1963linear}, which states that every two vertices of every polyhedron with~$n$ facets are connected via a path of~$\mathcal{O}(\mathrm{poly}(n))$ edges.
At this point, the best known pivot rules are randomized and achieve subexponential running times in expectation \cite{friedmann2011subexponential, hansen2015improved, kalai1992subexponential, matouvsek1992subexponential}.

For the most natural, memoryless and deterministic, pivot rules, exponential worst-case examples based on distorted hypercubes were constructed early on~\cite{avis1978notes, goldfarb1979worst, jeroslow1973simplex, klee1972good, murty1980computational}.
Amenta and Ziegler~\cite{amenta1999deformed} introduced the notion of deformed products to unify several of these constructions.
However, while this unification defines a class of polytopes that generalizes distorted hypercubes, it does not yield a unified exponential worst-case construction to exclude all pivot rules based on these deformed products, and neither does it yield new lower bounds for additional pivot rules.

Randomized and history-based pivot rules resisted similar approaches, and it was a major breakthrough in 2011 when Friedmann et al.~were able to prove the first subexponential lower bound for several randomized pivot rules~\cite{friedmann2011exponential, friedmann2011subexponential, hansen2012worst}.
They introduced a new technique based on a connection~\cite{puterman1994markov} between Howard's policy iteration algorithm~\cite{howard1960dynamic} for Markov decision processes (MDPs) and the simplex algorithm for linear programs (LPs).
The same technique was later used to prove exponential lower bounds for history-based pivot rules that had been candidates for polynomial time rules for a long time~\cite{avis2017exponential, disser2023exponential}.
While the approach via MDPs has proven powerful, the resulting analyses are often very technical (the full version of \cite{disser2023exponential} with all details of the proof has 197 pages).

In this paper, we apply the MDP-based technique to classical (memoryless and deterministic) pivot rules and obtain a unified construction that excludes several pivot rules at the same time, and any combination of them, while being relatively simple.

\subparagraph*{Our results.}

We give a unified worst-case construction for the policy iteration algorithm for MDPs that simultaneously applies to three of the most classical pivot rules.
The rigorous analysis of the resulting MDPs is reasonably concise.
We note that the exponential lower bounds for Dantzig's rule and the Largest Increase rule seem novel for the considered version of the policy iteration algorithm, while the result for Bland's rule is known~\cite{melekopoglou1994complexity}.

\begin{restatable}{theorem}{mainthm}\label{thm:main}
    There is a family~$(\mdz)_{n\in\N}$ of Markov decision processes~$\mdz$ of size~$\mathcal O(n)$ such that policy iteration performs~$\Omega(2^n)$ improving switches with Dantzig's rule, Bland's rule, and the Largest Increase pivot rule.
\end{restatable}

In fact, all three pivot rules apply the same set of improving switches with only slight differences in the order in which they get applied.
Because of this, the result still holds if we allow to change pivot rules during the course of the algorithm.

\begin{restatable}{corollary}{corcomb}\label{cor:comb}
	For any (deterministic or randomized) combination of Dantzig's, Bland's, or the Largest Increase rule, the policy iteration algorithm  has an exponential running time.
\end{restatable}

A well-known connection between policy iteration and the simplex method, allows to immediately translate our result to the simplex algorithm with the same pivot rules.
In particular, we obtain an exponential lower bound construction that holds even if, in every step, the entering variable is selected independently according to Dantzig's rule, Bland's rule, or the Largest Increase pivot rule, i.e., even if we change pivot rules during the course of the algorithm.
In other words, we obtain a lower bound for a family of pivot rules that results from combining these three rules.

\begin{restatable}{corollary}{cormain}\label{cor:main}
	There is a family~$(\mathcal{L}_n)_{n\in\N}$ of linear programs~$\mathcal{L}_n$ of size~$\mathcal{O}(n)$ such that the simplex algorithm performs~$\Omega(2^n)$ pivot operations for any (deterministic or randomized) combination of Dantzig's, Bland's, or the Largest Increase pivot rule.
\end{restatable}

\subparagraph*{Related work.}

Policy iteration for MDPs has been studied extensively for a variety of pivot rules.
In its original version~\cite{howard1960dynamic}, the algorithm applies improving switches to the current policy in all states simultaneously in every step.
Fearnley~\cite{Fearnley10} showed an exponential lower bound for a greedy pivot rule that selects the best improvement in every switchable state.
In this paper, we focus on pivot rules that only apply a single switch in each iteration.
Most of the MDP constructions for randomized or history-based pivot rules \cite{avis2017exponential,disser2023exponential,friedmann2011subexponential,hansen2012worst} consider this case, and Melekopoglou and Condon \cite{melekopoglou1994complexity} gave exponential lower bounds for several such deterministic pivot rules.
We emphasize that their constructions already include an exponential lower bound for Bland's rule~\cite{bland1977new}.
Since policy iteration is traditionally considered with simultaneous switches, to the best of our knowledge, no exponential lower bounds are known for Dantzig's rule~\cite{dantzig1951maximization} and the Largest Increase rule~\cite{dantzig1963linear} in the setting of single switches.

There is a strong connection between policy iteration and the simplex algorithm, which, under certain conditions (see below), yields that worst-case results for policy iteration carry over to the simplex method~\cite{puterman1994markov}.
This connection was used to establish subexponential lower bounds for randomized pivot rules, namely Randomized Bland~\cite{hansen2012worst} and Random-Edge, RaisingTheBar and Random-Facet~\cite{friedmann2011subexponential}.
It also lead to exponential lower bounds for history-based rules, namely Cunningham's rule~\cite{avis2017exponential} and Zadeh's rule~\cite{disser2023exponential}.
Conversely, lower bounds for the simplex algorithm with classical pivot rules were obtained via deformed hypercubes~\cite{amenta1999deformed} and do not transfer to MDPs.
Such results include lower bounds for Dantzig's rule~\cite{klee1972good}, the Largest Increase rule~\cite{jeroslow1973simplex}, Bland's rule~\cite{avis1978notes}, the Steepest Edge rule~\cite{goldfarb1979worst}, and the Shadow Vertex rule~\cite{murty1980computational}.
We provide an alternate lower bound construction for the first three of these rules via a family of MDPs.
As far as we can tell, as a side product, this yields the first exponential lower bound for policy iteration with Dantzig's rule and the Largest Increase rule.

While it remains open whether LPs can be solved in strongly polynomial time, there are several, both deterministic \cite{karmarkar1984new, khachiyan1980polynomial} and randomized \cite{bertsimas2004solving, dunagan2004simple, kelner2006randomized}, algorithms that solve LPs in weakly polynomial time.
A (strongly) polynomial time pivot rule for the simplex algorithm would immediately yield a strongly polynomial algorithm.

There have been different attempts to deal with the worst-case behavior of the simplex method from a theoretical perspective.
For example, the excessive running time was justified by showing that the simplex algorithm with Dantzig's original pivot rule is \emph{NP-mighty} \cite{disser2018simplex}, which means that it can be used to solve NP-hard problems.
This result was subsequently strengthened by establishing that deciding which solution is computed and whether a given basis will occur is PSPACE-complete~\cite{adler2014simplex, fearnley2015complexity}.
Recently, Disser and Mosis~\cite{disser2025unconditional} proved an exponential lower bound on the running time of the active-set method, a natural generalization of the simplex method to non-linear objectives, which holds \emph{for all pivot rules}.
The approach was refined to yield an unconditional lower bound for the active-set method in convex quadratic maximization~\cite{bach2025unconditionallowerboundactiveset}.
On the positive side, there are different results explaining the efficiency of the simplex method in practice, such as average-case analyses~\cite{adler1987simplex, borgwardt1982average, todd1986polynomial}.
Spielman and Teng \cite{spielman2004smoothed} introduced smoothed analysis as a way of bridging the gap between average-case and worst-case analysis.
They showed that the simplex algorithm with the shadow vertex pivot rule \cite{gass1955computational} has a polynomial smoothed complexity, and their results were further improved later \cite{dadush2018friendly, deshpande2005improved, huiberts2023upper, vershynin2009beyond}.

Another approach to derive stronger lower bounds on pivot rules is to consider combinatorical abstractions of LPs, such as Unique Sink Orientations (USOs)~\cite{gartner2006linear}. 
There is still a large gap between the best known deterministic algorithm for finding the unique sink, which is exponential~\cite{szabo2001unique}, and the almost quadratic lower bound~\cite{schurr2004finding}.
Considering randomized rules, the Random-Facet pivot rule, which is the best known simplex rule~\cite{hansen2015improved}, is also the best known pivot rule for acyclic USOs~\cite{gartner2002random}, achieving a subexponential running time in both settings.

\section{Preliminaries}
\subsection*{Markov Decision Processes}\label{sectionMDPs}

A \emph{Markov decision process} is an infinite duration one-player game on a finite directed graph~$G=(V_A,V_R,E_A,E_R,r,p)$.
The vertex set~$V=V_A\cup V_R$ of the graph is divided into \emph{agent vertices}~$V_A$ and \emph{randomization vertices}~$V_R$. 
Every agent edge~$e\in E_A\subseteq V_A \times V$ is assigned a \emph{reward}~$r(e)\in\R$, while every randomization edge~$\hat{e}\in E_R\subseteq V_R\times V$ is assigned a \emph{transition probability}~$p(\hat{e})\in[0,1]$. 
Outgoing transition probabilities add to one in every randomization vertex.

A process starts in an arbitrary starting vertex. If this is an agent vertex, the agent moves along one of the outgoing edges of this vertex (we assume that all vertices have at least one outgoing edge) and collects the corresponding reward. Otherwise, it gets randomly moved along one of the outgoing edges according to the transition probabilities. The process continues in this manner ad infinitum.

An agent vertex~$s\in V_A$ whose only outgoing edge is a self-loop with reward zero is called \emph{sink} of~$G$ if it is reachable from all vertices. 
A \emph{policy} for~$G$ is a function~$\pi \colon V_A\to V$ with~$(v,\pi(v))\in E_A$ for all~$v\in V_A$, determining the behavior of the process in agent vertices.
A policy~$\pi$ for~$G$ is called \emph{weak unichain} if~$G$ has a sink~$s$ such that~$\pi$ reaches~$s$ with a probability of one from every starting vertex.

The \emph{value} of a vertex~$v$ w.r.t.\ a policy~$\pi$ for a Markov decision process~$G$ is given by the expected total reward that the agent collects with policy~$\pi$ when the process starts in~$v$. 
More formally, the value function~$\Valsolo_{\pi,G}\colon V\to \R$ is defined by the following system of Bellman \cite{bellman1957dynamic} equations
    \begin{align*}
        \Valsolo_{\pi,G} (u) & = \begin{cases} \ r((u,\pi(u)))+\Valsolo_{\pi,G} (\pi(u)), &\text{if } u\in V_A, \\ \ \sum\limits_{v\in \Gamma^+(u)}p((u,v))\Valsolo_{\pi,G} (v), &\text{if } u\in V_R,\end{cases}
    \end{align*}
together with~$\Valsolo_{\pi,G}(s)=0$ if~$G$ has a sink~$s$.
The policy~$\pi$ is optimal (w.r.t.\ the \emph{expected total reward criterion}) if~$\Valsolo_{\pi,G}(v)\geq\Valsolo_{\tilde\pi,G}(v)$ for all~$v\in V_A$ and all policies~$\Tilde{\pi}$ for~$G$.
Whenever the underlying process~$G$ is clear from the context, we write~$\Val$ instead of~$\Valsolo_{\pi,G}$.

We say that the agent edge~$(u,v)\in E_A$ is an \emph{improving switch} for the policy~$\pi$ for process~$G$ if it satisfies~$z_{\pi,G}(u,v)\coloneqq r((u,v))+\Valsolo_{\pi,G}(v)-\Valsolo_{\pi,G}(u)>0$, where~$z_{\pi,G}(u,v)$ are the \emph{reduced costs} of~$(u,v)$ with respect to~$\pi$.
Again, we usually write~$z_\pi$ instead of~$z_{\pi,G}$.

If we \emph{apply} an improving switch~$s=(u,v)\in E_A$ to a policy~$\pi$, we obtain a new policy~$\pi^s$ which is given by~$\pi^s(u)= v$ and~$\pi^s(w)=\pi(w)$ for all~$w\in V_A\setminus \{u\}$. The improving switch~$s$ increases the value of~$u$ without decreasing the value of any other vertex.

\subsection*{Policy Iteration for Markov Decision Processes}

Howard's \cite{howard1960dynamic} policy iteration algorithm receives as input a finite Markov decision process~$G$ and a weak unichain policy~$\pi$ for~$G$. It then iteratively applies a set of improving switches to the current policy until there are none left. In the remainder of this paper, we consider a version of this algorithm that applies a single switch in every iteration, cf. Algorithm~\ref{alg:PI}. Due to monotonicity of the vertex values, this procedure visits every policy at most once. As there are only finitely many policies, the algorithm thus terminates after a finite number of iterations for every initial policy. 

~

\begin{algorithm}[H]
    \label{alg:PI}\caption{$\textsc{PolicyIteration}(G,\pi)$}
    \vgap
    \textbf{input: }a weak unichain policy~$\pi$ for a Markov decision process~$G$
    \vgap
    \hrule  
    \vgap
    \While{$\pi$ admits an improving switch}
    {
       ~$\bar s\gets \text{improving switch for } \pi$\;
        \vgap
       ~$\pi\gets \nwstr{\bar s}$
    }
    \textbf{return}~$\pi$
\end{algorithm}

~

We know that the policy iteration algorithm returns an optimal policy if there is an optimal policy which is weak unichain.

\begin{theorem}[\cite{friedmann2011exponential}]\label{thm: PI visits only wu}
    Let~$\pi$ be a weak unichain policy for a Markov decision process~$G$. If~$G$ admits a weak unichain, optimal policy, then~$\textsc{PolicyIteration}(G,\pi)$ only visits weak unichain policies and returns an optimal policy w.r.t.\ the expected total reward criterion.
\end{theorem}

In this paper, we consider the following three \emph{pivot rules}, i.e., rules that determine the choice of~$\textsc{PolicyIteration}(G,\pi)$ in each iteration:
\begin{itemize}
	\item \emph{Bland's pivot rule} assigns a unique number to every agent edge of~$G$. Then, in every iteration, it chooses the improving switch with the smallest number. 

	\item \emph{Dantzig's pivot rule} chooses an improving switch~$\bar s$ maximizing the reduced costs~$z_{\pi}(\bar s)$. 

	\item The \emph{Largest Increase rule} chooses an improving switch~$\bar s$ maximizing~$\sum_{v\in V_A}\Valsolo_{\nwstr{\bar s}}(v)$. 
\end{itemize}

\subsection*{A Connection between Policy Iteration and the Simplex Method}

Given a Markov decision process, we can formulate a linear program such that the application of the simplex method is in some sense equivalent to the application of policy iteration.
We refer to~\cite{friedmann2011exponential} for more details and the derivation of the following result.

\begin{theorem}[\cite{friedmann2011exponential}]\label{thm:the connection}
    Let~$\pi$ be a weak unichain policy for a Markov decision process~$G$. Assume that there is an optimal, weak unichain policy for~$G$ and that~$\textsc{PolicyIteration}(G,\pi)$ with a given pivot rule takes~$N$ iterations. Then, there is an LP of linear size such that the simplex algorithm with the same pivot rule takes~$N$ iterations.
\end{theorem}

In terms of the simplex method, Bland's pivot rule chooses the entering variable of smallest index~\cite{bland1977new}, Dantzig's rule chooses an entering variable maximizing the reduced costs~\cite{dantzig1951maximization}, and the Largest Increase rule greedily maximizes the objective function value.

The linear program in the previous theorem has one variable for every agent edge of the Markov decision process such that the reduced costs of a given edge equal the reduced costs of the corresponding variable, and the objective function equals the sum over all vertex values as given in the Largest Increase rule for policy iteration~\cite{avis2017exponential,disser2023exponential,hansen2012worst}.
Therefore, the choices of each pivot rule in the two settings are consistent.

Additionally, we want to mention that the linear program from Theorem~\ref{thm:the connection} is always non-degenerate.
Therefore, we cannot reduce the number of required iterations on these programs by combining a given pivot rule with the Lexicographic pivot rule \cite{dantzig1955generalized}.

\subsection*{Notation}

Let~$n\in\N$ be fixed.
We write~$[n]=\{1,2,\dots,n\}$ and~$[n]_0=\{0,1,\dots,n\}$.
Then, the set of all numbers that can be represented with~$n$ bits is~$[2^n-1]_0$.

For every~$x\in[2^n-1]_0$ and $i\in[n]$, let~$x_i$ denote the~$i$-th bit of~$x$, i.e.,~$x=\sum_{i\in[n]}x_i2^{i-1}$, and let~$\Lx ix=\max\{j\in[i-2]\mid x_j=1\text{ or } j=1\}$ for~$i\geq 3$.
Finally, for~$x\in[2^n-1]$, we denote the \emph{least significant set bit}  of~$x$ by~$\ell_1(x)=\min\{i\in[n]:x_i= 1\}$, and the \emph{most significant set bit} of~$x$ by~$\m(x)=\max\{i\in[n]:x_i= 1\}$.

Let~$G=(V_A,V_R,E_A,E_R,r,p)$ be a Markov decision process. 
For~$v\in V_A\cup V_R$, we write~$\Gamma^+_G(v)=\{w\in V_A\cup V_R\colon(v,w)\in E_A\cup E_R\}$.
If the underlying process is clear from the context, we just write~$\Gamma^+(v)$.

\section{An Exponential Lower Bound for Bland's pivot rule}
\label{sec:Bland}

In this section, we consider a family~$(\mbl=(V_{\mbl},E_{\mbl},r_{\mbl}))_{n\in\N}$ of Markov decision processes, which do not involve any randomization.
Consider Figure~\ref{fig:subfig1} for a drawing of~$\mathcal{B}_4$.
Every process~$\mbl$ consists of~$n$ separate levels, together with a global \emph{transportation vertex}~$t$, a sink~$s$, and a \emph{dummy vertex}~$d$.
Each level~$\ell\in[n]$ comprises two vertices, called~$a_\ell$ and~$b_\ell$.
For convenience, we sometimes denote the sink by~$a_{n+1}$ and the dummy vertex by~$b_{n+1}$.

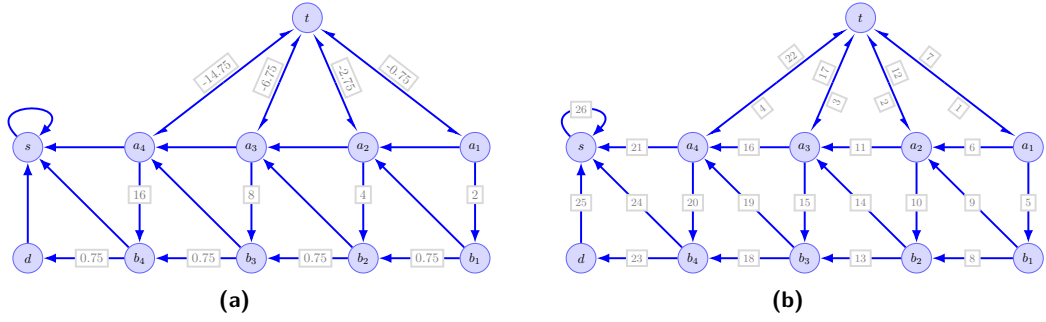
\begin{figure}
  \centering
  \begin{subfigure}{0.48\textwidth}
    \centering
    \captionsetup{justification=centering}
    \resizebox{\linewidth}{!}{%
    \begin{tikzpicture}[node distance = 3cm,
      p0/.style={circle, draw=blue!60, fill=blue!15, thick, minimum size=8mm},
      bl/.style={fill=white,draw=gray!30,text=gray},
      every path/.style={->,draw=blue,line width=1.5pt,{>=Latex}}]
    
      \node[p0] (A) {$a_1$};
      \node[p0,below of=A] (B) {$b_1$};
      \draw (A) -- node [bl,pos=0.4] {2} (B);
      
      \node[p0,left of=A] (C) {$a_2$};
      \node[p0,left of=B] (D) {$b_2$};
      \draw (A) -- (C);
      \draw (B) -- node [bl,pos=0.4] {0.75} (D);
      \draw (B) -- (C);
      \draw (C) -- node [bl,pos=0.4] {4} (D);
      
      \node[p0,left of=C] (E) {$a_3$};
      \node[p0,left of=D] (F) {$b_3$};
      \draw (C) -- (E);
      \draw (D) -- (E);
      \draw (D) -- node [bl,pos=0.4] {0.75} (F);
      \draw (E) -- node [bl,pos=0.4] {8} (F);

      \node[p0,left of=E] (G) {$a_4$};
      \node[p0,left of=F] (H) {$b_4$};
      \draw (E) -- (G);      
      \draw (F) -- (G);
      \draw (F) -- node [bl,pos=0.4] {0.75} (H);
      \draw (G) -- node [bl,pos=0.4] {16} (H);

      \node[p0,left of=G] (I) {$s$};
      \node[p0, left of =H] (J) {$d$};
      \draw (G) -- (I);  
      \draw (H) -- (I);
      \draw (H) -- node [bl,pos=0.4] {0.75} (J);      
      \draw (J) -- (I);
      \draw[out=135,in=45, loop, looseness=6] (I) edge  (I) {};

      \node[p0] at (-4.5,3.5) (K) {$t$};
      \draw[arrows = {->[harpoon,swap]}] ($(A)!0.5!(K)$) -- node [pos=0, bl,sloped,above=.2em] {-0.75} (K);
      \draw[arrows = {->[harpoon,swap]}] ($(K)!0.5!(A)$) -- node [] {} (A);

      \draw[arrows = {->[harpoon,swap]}] ($(C)!0.5!(K)$) -- node [pos=0, bl,sloped,above=.2em] {-2.75} (K);
      \draw[arrows = {->[harpoon,swap]}] ($(K)!0.5!(C)$) -- node [] {} (C);

      \draw[arrows = {->[harpoon]}] ($(E)!0.5!(K)$) -- node [pos=0, bl,sloped,above=.2em] {-6.75} (K);
      \draw[arrows = {->[harpoon]}] ($(K)!0.5!(E)$) -- node [] {} (E);

      \draw[arrows = {->[harpoon]}] ($(G)!0.5!(K)$) -- node [pos=0, bl,sloped,above=.2em] {-14.75} (K);
      \draw[arrows = {->[harpoon]}] ($(K)!0.5!(G)$) -- node [] {} (G);

    \end{tikzpicture}%
    }
    \caption{}
    \label{fig:subfig1}
  \end{subfigure}
  \hfill
  \begin{subfigure}{0.48\textwidth}
    \centering
    \captionsetup{justification=centering}
    \resizebox{\linewidth}{!}{%
    \begin{tikzpicture}[node distance = 3cm,
      p0/.style={circle, draw=blue!60, fill=blue!15, thick, minimum size=8mm},
      bl/.style={fill=white,draw=gray!30,text=gray,font=\footnotesize},
      every path/.style={->,draw=blue,line width=1.5pt,{>=Latex}}]
    
      \node[p0] (A) {$a_1$};
      \node[p0,below of=A] (B) {$b_1$};
      \draw (A) -- node [bl] {5} (B);
      
      \node[p0,left of=A] (C) {$a_2$};
      \node[p0,left of=B] (D) {$b_2$};
      \draw (A) -- node [bl] {6} (C);
      \draw (B) -- node [bl] {8} (D);
      \draw (B) -- node [bl] {9} (C);
      \draw (C) -- node [bl] {10} (D);
      
      \node[p0,left of=C] (E) {$a_3$};
      \node[p0,left of=D] (F) {$b_3$};
      \draw (C) -- node [bl] {11} (E);
      \draw (D) -- node [bl] {14} (E);
      \draw (D) -- node [bl] {13} (F);
      \draw (E) -- node [bl] {15} (F);

      \node[p0,left of=E] (G) {$a_4$};
      \node[p0,left of=F] (H) {$b_4$};
      \draw (E) -- node [bl] {16} (G);      
      \draw (F) -- node [bl] {19} (G);
      \draw (F) -- node [bl] {18} (H);
      \draw (G) -- node [bl] {20} (H);

      \node[p0,left of=G] (I) {$s$};
      \node[p0, left of =H] (J) {$d$};
      \draw (G) --  node [bl] {21} (I);  
      \draw (H) -- node [bl] {24} (I);
      \draw (H) -- node [bl] {23} (J);      
      \draw (J) -- node [bl] {25} (I);
      \draw[out=135,in=45, loop, looseness=6] (I) edge node [bl] {26} (I) {};

      \node[p0] at (-4.5,3.5) (K) {$t$};
      \draw[arrows = {->[harpoon,swap]}] ($(A)!0.5!(K)$) -- node [pos=0.3, bl,sloped,above=.2em] {7} (K);
      \draw[arrows = {->[harpoon,swap]}] ($(K)!0.5!(A)$) -- node [pos=0.3, bl,sloped,below=.2em] {1} (A);

      \draw[arrows = {->[harpoon,swap]}] ($(C)!0.5!(K)$) -- node [pos=0.1, bl,sloped,above=.2em] {12} (K);
      \draw[arrows = {->[harpoon,swap]}] ($(K)!0.5!(C)$) -- node [pos=0.3, bl,sloped,below=.2em] {2} (C);

      \draw[arrows = {->[harpoon]}] ($(E)!0.5!(K)$) -- node [pos=0.1, bl,sloped,above=.2em] {17} (K);
      \draw[arrows = {->[harpoon]}] ($(K)!0.5!(E)$) -- node [pos=0.3, bl,sloped,below=.2em] {3} (E);

      \draw[arrows = {->[harpoon]}] ($(G)!0.5!(K)$) -- node [pos=0.3, bl,sloped,above=.2em] {22} (K);
      \draw[arrows = {->[harpoon]}] ($(K)!0.5!(G)$) -- node [pos=0.3, bl,sloped,below=.2em] {4} (G);

    \end{tikzpicture}%
    }
    \caption{}
    \label{fig:subfig2}
  \end{subfigure}
  \caption{Two drawings of the Markov decision process~$\mathcal{B}_4$. In (a), edge labels denote rewards and unlabeled edges have a reward of zero. In (b), edge labels define the Bland numbering~$\mathcal{N}_{\mathcal{B}_4}$.}
  \label{fig:B4}
\end{figure}

In vertex~$a_\ell$, the agent can either \emph{enter} level~$\ell$ by going to vertex~$b_\ell$, \emph{skip} this level by going to vertex~$a_{\ell+1}$, or \emph{board} the transportation vertex by going to~$t$.
From the transportation vertex, the agent \emph{travels} to one of the vertices~$a_i$ with~$i\in[n]$.
In~$b_\ell$, the agent can decide between \emph{leaving} the set~$\bigcup_{i\in[n+1]}\{b_i\}$ by going to~$a_{\ell+1}$ and \emph{staying} in this set by going to~$b_{\ell+1}$.
We will simply say that the agent leaves level~$\ell$ or stays in level~$\ell$, respectively. 

Finally, when the agent reaches the dummy vertex~$d$, it must go to the sink, and the only outgoing edge of the sink~$s$ is the self-loop~$(s,s)$. 

The function~$r_{\mbl}$ grants the agent a reward of~$2^\ell$ for entering level~$\ell$, a reward of~$0.75$ for staying in level~$\ell$, and a (negative) reward of~$(-2^{\ell}+1.25)$ for boarding~$t$ from~$a_\ell$; all other rewards are zero. 

The \emph{Bland numbering}~$\order\colon E_{\mbl}\to |E_{\mbl}|$ of the edges of~$\mbl$ is defined in Table~\ref{table:Bland}, together with~$\order((d,s))=6n+1$ and~$\order((s,s))=6n+2(=|E_{\mbl}|)$. 
This table also contains alternative names for the edges, which match the description above and which we will use to simplify the exposition.
Consider Figure~\ref{fig:subfig2} for the Bland numbering of~$\mathcal{B}_4$.

\begin{table}[h]
    \centering
    \begin{tabular}{ |c|c|c| }
     \hline
     \multicolumn{2}{|c|}{$e\in E_{\mbl}$} & $\order(e)$ \\
     \hline\hline
     ($t$,$a_i$) & $\travel(i)$ & $i$ \\ 
     \hline
     $(a_i,b_i)$ & $\enter(i)$ & $n+1+5(i-1)$\\
     \hline
     $(a_i,a_{i+1})$ & $\skipp(i)$ & $n+2+5(i-1)$\\
     \hline
     $(a_i,t)$ & $\board(i)$ & $n+3+5(i-1)$\\
     \hline
     $(b_i,b_{i+1})$ & $\stay(i)$ & $n+4+5(i-1)$\\
     \hline
     $(b_i,a_{i+1})$ & $\leave(i)$ & $n+5+5(i-1)$\\
     \hline
\end{tabular}\vspace{.1cm}
    \caption{Edge names and the definition of the Bland numbering~$\order$, where~$i\in[n]$.}
    \label{table:Bland}
\end{table} 

In the following, consider~$\mbl$ for some arbitrary but fixed~$n\in\N$.
The aim of this section is to show that \textsc{PolicyIteration} with Bland's pivot rule, cf. Algorithm~\ref{alg:Bland}, applies~$\Omega(2^n)$ improving switches when given~$\mbl$, a suitable initial policy, and~$\order$ as input.

\begin{algorithm}[h]
    \caption{$\textsc{Bland}(G,\pi,\mathcal{N})$}\label{alg:Bland}
    \vgap
    \textbf{input: }Markov decision process~$G$, weak unichain policy~$\pi$, edge numbering~$\mathcal N$
    \vgap
    \hrule  
    \vgap
    \While{$\pi$ admits an improving switch}
    {
       ~$\bar s\gets \text{the improving switch~$s$ for~$\pi$ that minimizes~$\mathcal{N}(s)$}$\;
        \vgap
       ~$\pi\gets \nwstr{\bar s}$
    }
    \textbf{return}~$\pi$
\end{algorithm}

More precisely, we will see that the algorithm visits all of the following policies.

\begin{definition2}\label{canonical policy}
    The policy~$\pi_0$ for~$\mbl$ such that~$\travel(1)$ is active, and~$\skipp(i)$ and~$\leave(i)$ are active for all~$i\in[n]$ is the \emph{canonical policy} for~$0$.
    \mbox{For~$x\in[2^n-1]$}, the policy~$\pi_x$ for~$\mbl$ is the \emph{canonical policy} for~$x$ if it satisfies the following conditions:
    \begin{enumerate}[<1>]
        \item\label{1} The policy travels from~$t$ to the least significant set bit, i.e.,~$\travel(\lone(x))$ is active.
        \item It collects no reward above the most significant set bit, i.e.,~$\leave(\m(x))$,~$\skipp(i)$, and~$\leave(i)$ are active for all~$\m(x)<i\leq n$.
        \item Every set bit~$x_i=1$ determines the behavior of the policy down to the next, less significant set bit or, if~$i=\ell_1(x)$, down to the first bit:
        \begin{enumerate}[<a>]
            \item $\enter(i)$ is active.
            \item if~$i=2$, then~$\leave(1)$ is active. If additionally~$x_1=0$, then~$\skipp(1)$ is active.
            \item if~$i\geq 3$ and~$x_{i-1}=1$, then~$\leave(i-1)$ is active.
            \item if~$i\geq 3$ and~$x_{i-1}=0$: 
            \begin{enumerate}[<$\text{d}_1$>]
                \item $\stay(i-1)$,~$\skipp(i-1)$, and~$\leave(i-2)$ are active.
                \item if~$\Lx ix< i-2$, then for all~$j\in\{\Lx ix+1, \dots, i-2\}$, the edges~$\board(j)$ and~$\stay(j-1)$ are active; if~$\Lx ix=1$ and~$x_1=0$, then~$\board(1)$ is active.\lipicsEnd
            \end{enumerate}
        \end{enumerate}        
    \end{enumerate}
\end{definition2}

Consider Figure~\ref{fig:c7} and Figure~\ref{fig:c8} for examples of canonical policies.
Note that canonical policies exist and are unique as the definition contains precisely one condition on every agent vertex with more than one outgoing edge.
Further, the~$2^n$ canonical policies are pairwise different as~$\enter(i)$ is active in~$\pi_x$ if and only if~$x_i=1$. 

We will now analyze the behavior of~$\textsc{Bland}(\mbl,\pi_0,\order)$, i.e., we choose the canonical policy for zero as our initial policy.
Since this policy visits every vertex except the sink only once, it is weak unichain.

\begin{observation}\label{pi0 is wu for B}
The canonical policy~$\pi_0$ is a weak unichain policy for~$\mbl$.
\end{observation}

Thus, according to Theorem~\ref{thm:the connection}, the following result will allow us to transfer our results for the policy iteration algorithm to the simplex method.

\begin{restatable}{lemma}{optforB}\label{opt for B}
    Let the policy~$\pi_*$ for~$\mbl$ be determined as follows: 
   ~$\stay(n)$ and~$\travel(1)$ are active,~$\enter(i)$ is active for all~$i\in[n]$, and~$\leave(j)$ is active for all~$j\in[n-1]$.    
    Then,~$\pi_*$ is weak unichain and optimal for~$\mbl$.
\end{restatable}
\begin{proof}
Since~$\pi_*$ visits every vertex, besides the sink, only once, it is weak unichain.
For optimality, note that~$t$ travels to~$a_1$ and that, when starting in a vertex~$a_{\ell}$, \mbox{policy~$\pi_*$} enters level $\ell$ and all levels above and collects the reward of~$\stay(n)$.
The policy is thus clearly optimal among the set of policies that do not use boarding edges.

Further, we have~$r_{\mbl}(\board(\ell))=-2^{\ell}+1.25=-(\sum_{i=1}^{\ell-1}2^i+0.75)$. That is, the costs of~$\board(\ell)$ equal the maximum reward that can be collected in the first~$\ell-1$ levels.
Thus, we cannot increase vertex values by using boarding edges, which yields that~$\pi_*$ is optimal.
\end{proof}

The following technical result will be helpful in the upcoming proofs.

\begin{restatable}{lemma}{lemnopointimpr}\label{no pointer improving}
    Let~$x\in[2^n-1]_0$ and~$i\in[n]$. Then,~$\travel(i)$ is not improving for~$\pi_x$.
\end{restatable}
\begin{proof}
    All vertex values with respect to~$\pi_0$ are zero, and~$r_{\mbl}(\travel(i))=0$. 
    Thus, the claim holds for~$x=0$, so we assume~$x\in[2^n-1]$ in the following.
    
    Let the vertices~$a_k$ and~$a_\ell$ either correspond to successive set bits, i.e.,~$x_k=x_\ell=1$ and~$x_j=0$ for all~$k<j<\ell$, or let~$k=\m(x)$ and~$\ell= n+1$.
    Either way, \mbox{Definition \ref{canonical policy}} implies that~$\pi_x$ includes a path from~$a_k$ to~$a_\ell$, which does not contain any boarding edge.
    Hence, we have~$\Valsolo_{\pi_x}(a_\alpha)\geq\Valsolo_{\pi_x}(a_\beta)\geq0$ for all set bits~$x_\alpha= x_\beta=1$ with~$\alpha\leq\beta$.
    Since the transportation vertex chooses the least significant set bit in~$\pi_x$, this yields that~$\travel(i)$ is not improving if~$x_i=1$.

    Further, Definition \ref{canonical policy} yields that~$x_j=1$ if and only if~$\enter(j)$ is active in~$\pi_x$.
    Thus, when starting in some vertex~$a_i$ with~$x_i=0$, policy~$\pi_x$ either boards~$t$ from~$a_i$ or it skips levels until reaching a node that boards~$t$, a level corresponding to a set bit, or the sink.
    In all four cases,~$\travel(i)$ is not improving.
    This completes the proof.    
\end{proof}

We will show in two steps that, when using the initial policy~$\pi_0$, \textsc{Bland} visits all of the other canonical policies.
Firstly, given the canonical policy for an arbitrary even integer~$x$, we see that the algorithm applies improving switches until reaching the canonical policy~$\pi_{x+1}$.

\begin{restatable}{lemma}{evencanphases}\label{Bland applies even can phases}
Let~$x\in[2^n-2]_0$ be even.
Then,~$\textsc{Bland}(\mbl,\pi_x,\order)$ visits~$\pi_{x+1}$.
\end{restatable}
\begin{proof}
According to Lemma~\ref{no pointer improving}, no travel edges are improving for~$\pi_x$, so the Bland numbering~$\order$ yields that the algorithm applies the switch~$\enter(1)$ to~$\pi_x$ if it is improving.
This edge is improving for~$\pi_0$, and one can easily check that its application results in the canonical policy~$\pi_1$.
Hence, it suffices to consider~$x\neq 0$ in the following.
This yields~${\ell_1}\coloneqq\ell_1(x)>1$ as~$x$ is even. 
Influenced by condition <3> from Definition \ref{canonical policy}, we consider two cases.

Firstly, if~${\ell_1}=2$, conditions <1> and <b> state that~$\travel(2)$,~$\skipp(1)$ and~$\leave(1)$ are active in~$\pi_x$. 
Hence,~$\enter(1)$ is improving for~$\pi_x$ and gets applied by \textsc{Bland}.
The edge~$\travel(1)$ becomes improving and gets applied next as it minimizes~$\order$.

Secondly, if~${\ell_1}\geq3$, conditions <1> and <d> yield that~$\pi_x$ includes the paths~$(a_1,t,a_{{\ell_1}})$ and~$(b_1, b_2, \ldots,b_{{\ell_1}-2},a_{{\ell_1}-1},a_{{\ell_1}})\eqqcolon P$. 
Hence,~$\Valsolo_{\pi_x}(a_1)\leq\Valsolo_{\pi_x}(a_{{\ell_1}})\leq\Valsolo_{\pi_x}(b_1)$.
Therefore, as~$\enter(1)$ has a positive reward, it is improving and gets applied to~$\pi_x$.
The new policy walks from~$a_1$ to~$b_1$ and then follows the path~$P$, so~$a_1$ has a higher value than~$a_{{\ell_1}}$. Since~$\travel({\ell_1})$ is active in~$\pi_x$, the edge~$\travel(1)$ is improving and gets applied next.

Let~$\pi$ denote the policy resulting from the application of~$\enter(1)$ and~$\travel(1)$ to~$\pi_x$.
It now suffices to show that~$\pi$ satisfies the conditions of Definition~\ref{canonical policy} for~$x+1$.

As~$x+1$ is odd, we have~$\ell_1(x+1)=1$, so~$\pi$ satisfies the first condition.
Further, the second condition remains satisfied as both applied switches are below the most significant set bit.
Finally, the third condition now requires that~$\enter(1)$ is active~-- instead of~$\skipp(1)$ if~${\ell_1}=2$, or~$\board(1)$ if~${\ell_1}\geq3$~-- and otherwise contains the same requirements.
Hence,~$\pi$ is the canonical policy for~$x+1$.
\end{proof}

Secondly, we need to show that the algorithm also transforms the canoncial policy~$\pi_x$ for an arbitrary odd number~$x$ into the next canonical policy~$\pi_{x+1}$.
We will see that the algorithm does this by applying the following sequence of improving switches.

\begin{definition2}\label{canonical phases odd}
Let~$x\in[2^n-3]$ be odd and write~$\ell\coloneqq \ell_0(x)>1$. 
Then, the \emph{canonical phases} with respect to~$x$ are:
    \begin{enumerate}[1.]
    \setlength\itemsep{.1em}
        \item If~$x_{\ell+1}=1$, activate~$\leave(\ell)$.
        \item If~$x_{\ell+1}=1$ or~$\ell>\m(x)$, activate~$\stay(\ell-1)$.
        \item Activate~$\enter(\ell)$ and~$\travel(\ell)$.
        \item If~$\ell\geq 3$, activate~$\board(j)$ for all~$j\in[\ell-2]$ in increasing order.
        \item Activate~$\skipp(\ell-1)$.
        \item If~$\ell\geq 4$, activate~$\stay(j)$ for all~$j\in[\ell-3]$ in decreasing order.
        \item If~$\ell=2$, activate~$\leave(1)$.\lipicsEnd
    \end{enumerate}
\end{definition2}

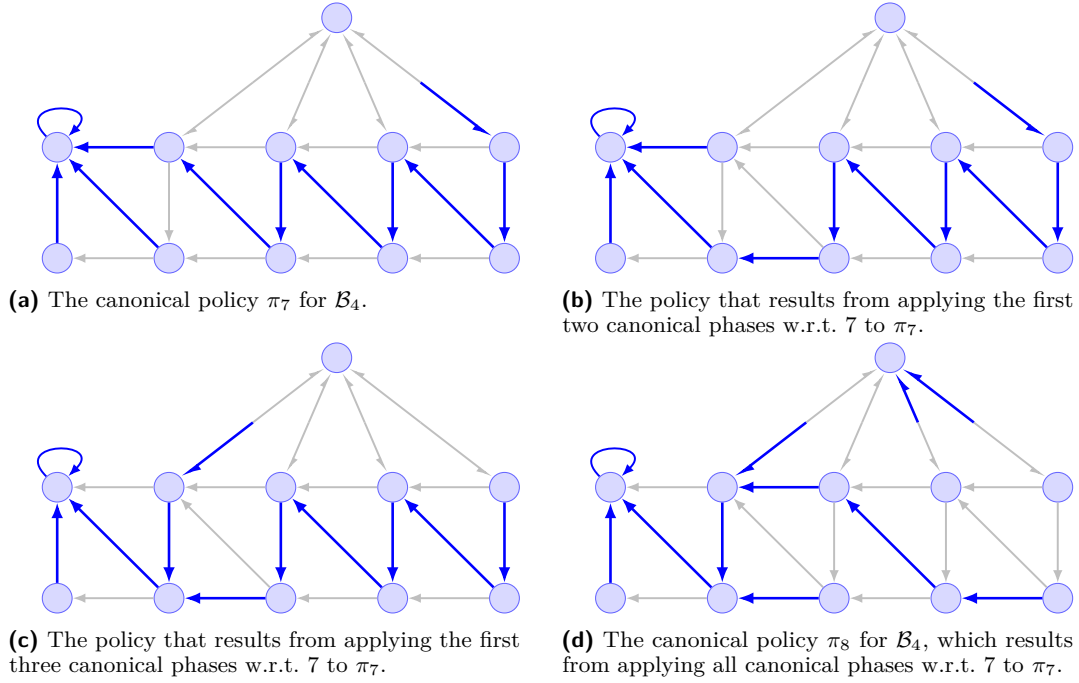
\begin{figure}
\centering
  \begin{subfigure}{0.48\textwidth}
    \centering
    \resizebox{\textwidth}{!}{%
    \begin{tikzpicture}[node distance = 3cm,
      p0/.style={circle, draw=blue!60, fill=blue!15, thick, minimum size=8mm},
      bl/.style={fill=white,draw=gray!30,text=gray},
      every path/.style={->,draw=blue,line width=1.5pt,{>=Latex}}]
    
      \node[p0] (A) {};
      \node[p0,below of=A] (B) {};
      \draw[line width=2pt] (A) -- (B);
      
      \node[p0,left of=A] (C) {};
      \node[p0,left of=B] (D) {};
      \draw[lightgray] (A) -- (C);
      \draw[lightgray] (B) -- (D);
      \draw[line width=2pt] (B) -- (C);
      \draw[line width=2pt] (C) -- (D);
      
      \node[p0,left of=C] (E) {};
      \node[p0,left of=D] (F) {};
      \draw[lightgray] (C) -- (E);
      \draw[line width=2pt] (D) -- (E);
      \draw[lightgray] (D) -- (F);
      \draw[line width=2pt] (E) -- (F);

      \node[p0,left of=E] (G) {};
      \node[p0,left of=F] (H) {};
      \draw[lightgray] (E) -- (G);      
      \draw[line width=2pt] (F) -- (G);
      \draw[lightgray] (F) -- (H);
      \draw[lightgray] (G) -- (H);

      \node[p0,left of=G] (I) {};
      \node[p0, left of =H] (J) {};
      \draw[line width=2pt] (G) -- (I);  
      \draw[line width=2pt] (H) -- (I);
      \draw[lightgray] (H) -- (J);      
      \draw[line width=2pt] (J) -- (I);
      \draw[out=135,in=45, loop, looseness=6,line width=2pt] (I) edge  (I) {};

      \node[p0] at (-4.5,3.5) (K) {};
      \draw[arrows = {->[harpoon,swap]},lightgray] ($(A)!0.5!(K)$) -- (K);
      \draw[arrows = {->[harpoon,swap]},blue,line width=2pt] ($(K)!0.5!(A)$) -- node [] {} (A);

      \draw[arrows = {->[harpoon,swap]},lightgray] ($(C)!0.5!(K)$) -- (K);
      \draw[arrows = {->[harpoon,swap]},lightgray] ($(K)!0.5!(C)$) -- node [] {} (C);

      \draw[arrows = {->[harpoon]},lightgray] ($(E)!0.5!(K)$) -- (K);
      \draw[arrows = {->[harpoon]},lightgray] ($(K)!0.5!(E)$) -- node [] {} (E);

      \draw[arrows = {->[harpoon]},lightgray] ($(G)!0.5!(K)$) -- (K);
      \draw[arrows = {->[harpoon]},lightgray] ($(K)!0.5!(G)$) -- node [] {} (G);

    \end{tikzpicture}%
    }
    \caption{The canonical policy~$\pi_7$ for~$\mathcal{B}_4$.\\~}
    \label{fig:c7}
  \end{subfigure}
  \hfill
  \begin{subfigure}{0.48\textwidth}
    \centering
    \resizebox{\textwidth}{!}{%
\begin{tikzpicture}[node distance = 3cm,
      p0/.style={circle, draw=blue!60, fill=blue!15, thick, minimum size=8mm},
      bl/.style={fill=white,draw=gray!30,text=gray},
      every path/.style={->,draw=blue,line width=1.5pt,{>=Latex}}]
    
      \node[p0] (A) {};
      \node[p0,below of=A] (B) {};
      \draw[line width=2pt] (A) -- (B);
      
      \node[p0,left of=A] (C) {};
      \node[p0,left of=B] (D) {};
      \draw[lightgray] (A) -- (C);
      \draw[lightgray] (B) -- (D);
      \draw[line width=2pt] (B) -- (C);
      \draw[line width=2pt] (C) -- (D);
      
      \node[p0,left of=C] (E) {};
      \node[p0,left of=D] (F) {};
      \draw[lightgray] (C) -- (E);
      \draw[line width=2pt] (D) -- (E);
      \draw[lightgray] (D) -- (F);
      \draw[line width=2pt] (E) -- (F);

      \node[p0,left of=E] (G) {};
      \node[p0,left of=F] (H) {};
      \draw[lightgray] (E) -- (G);      
      \draw[lightgray] (F) -- (G);
      \draw[line width=2pt] (F) -- (H);
      \draw[lightgray] (G) -- (H);

      \node[p0,left of=G] (I) {};
      \node[p0, left of =H] (J) {};
      \draw[line width=2pt] (G) -- (I);  
      \draw[line width=2pt] (H) -- (I);
      \draw[lightgray] (H) -- (J);      
      \draw[line width=2pt] (J) -- (I);
      \draw[out=135,in=45, loop, looseness=6,line width=2pt] (I) edge  (I) {};

      \node[p0] at (-4.5,3.5) (K) {};
      \draw[arrows = {->[harpoon,swap]},lightgray] ($(A)!0.5!(K)$) -- (K);
      \draw[arrows = {->[harpoon,swap]},blue,line width=2pt] ($(K)!0.5!(A)$) -- node [] {} (A);

      \draw[arrows = {->[harpoon,swap]},lightgray] ($(C)!0.5!(K)$) -- (K);
      \draw[arrows = {->[harpoon,swap]},lightgray] ($(K)!0.5!(C)$) -- node [] {} (C);

      \draw[arrows = {->[harpoon]},lightgray] ($(E)!0.5!(K)$) -- (K);
      \draw[arrows = {->[harpoon]},lightgray] ($(K)!0.5!(E)$) -- node [] {} (E);

      \draw[arrows = {->[harpoon]},lightgray] ($(G)!0.5!(K)$) -- (K);
      \draw[arrows = {->[harpoon]},lightgray] ($(K)!0.5!(G)$) -- node [] {} (G);

    \end{tikzpicture}%
    }
    \caption{The policy that results from applying the first two canonical phases w.r.t.\ 7 to~$\pi_7$.}
    \label{fig:c71}
  \end{subfigure}
  \hfill
  \begin{subfigure}{0.48\textwidth}
    \centering
    \resizebox{\textwidth}{!}{%
\begin{tikzpicture}[node distance = 3cm,
      p0/.style={circle, draw=blue!60, fill=blue!15, thick, minimum size=8mm},
      bl/.style={fill=white,draw=gray!30,text=gray},
      every path/.style={->,draw=blue,line width=1.5pt,{>=Latex}}]
    
      \node[p0] (A) {};
      \node[p0,below of=A] (B) {};
      \draw[line width=2pt] (A) -- (B);
      
      \node[p0,left of=A] (C) {};
      \node[p0,left of=B] (D) {};
      \draw[lightgray] (A) -- (C);
      \draw[lightgray] (B) -- (D);
      \draw[line width=2pt] (B) -- (C);
      \draw[line width=2pt] (C) -- (D);
      
      \node[p0,left of=C] (E) {};
      \node[p0,left of=D] (F) {};
      \draw[lightgray] (C) -- (E);
      \draw[line width=2pt] (D) -- (E);
      \draw[lightgray] (D) -- (F);
      \draw[line width=2pt] (E) -- (F);

      \node[p0,left of=E] (G) {};
      \node[p0,left of=F] (H) {};
      \draw[lightgray] (E) -- (G);      
      \draw[lightgray] (F) -- (G);
      \draw[line width=2pt] (F) -- (H);
      \draw[line width=2pt] (G) -- (H);

      \node[p0,left of=G] (I) {};
      \node[p0, left of =H] (J) {};
      \draw[lightgray] (G) -- (I);  
      \draw[line width=2pt] (H) -- (I);
      \draw[lightgray] (H) -- (J);      
      \draw[line width=2pt] (J) -- (I);
      \draw[out=135,in=45, loop, looseness=6,line width=2pt] (I) edge  (I) {};

      \node[p0] at (-4.5,3.5) (K) {};
      \draw[arrows = {->[harpoon,swap]},lightgray] ($(A)!0.5!(K)$) -- (K);
      \draw[arrows = {->[harpoon,swap]},lightgray] ($(K)!0.5!(A)$) -- node [] {} (A);

      \draw[arrows = {->[harpoon,swap]},lightgray] ($(C)!0.5!(K)$) -- (K);
      \draw[arrows = {->[harpoon,swap]},lightgray] ($(K)!0.5!(C)$) -- node [] {} (C);

      \draw[arrows = {->[harpoon]},lightgray] ($(E)!0.5!(K)$) -- (K);
      \draw[arrows = {->[harpoon]},lightgray] ($(K)!0.5!(E)$) -- node [] {} (E);

      \draw[arrows = {->[harpoon]},lightgray] ($(G)!0.5!(K)$) -- (K);
      \draw[arrows = {->[harpoon]},blue,line width=2pt] ($(K)!0.5!(G)$) -- node [] {} (G);

    \end{tikzpicture}%
    }
    \caption{The policy that results from applying the first three canonical phases w.r.t.\ 7 to~$\pi_7$.}
    \label{fig:c72}
  \end{subfigure}
  \hfill
  \begin{subfigure}{0.48\textwidth}
    \centering
    \resizebox{\textwidth}{!}{%
\begin{tikzpicture}[node distance = 3cm,
      p0/.style={circle, draw=blue!60, fill=blue!15, thick, minimum size=8mm},
      bl/.style={fill=white,draw=gray!30,text=gray},
      every path/.style={->,draw=blue,line width=1.5pt,{>=Latex}}]
    
      \node[p0] (A) {};
      \node[p0,below of=A] (B) {};
      \draw[lightgray] (A) -- (B);
      
      \node[p0,left of=A] (C) {};
      \node[p0,left of=B] (D) {};
      \draw[lightgray] (A) -- (C);
      \draw[line width=2pt] (B) -- (D);
      \draw[lightgray] (B) -- (C);
      \draw[lightgray] (C) -- (D);
      
      \node[p0,left of=C] (E) {};
      \node[p0,left of=D] (F) {};
      \draw[lightgray] (C) -- (E);
      \draw[line width=2pt] (D) -- (E);
      \draw[lightgray] (D) -- (F);
      \draw[lightgray] (E) -- (F);

      \node[p0,left of=E] (G) {};
      \node[p0,left of=F] (H) {};
      \draw[line width=2pt] (E) -- (G);      
      \draw[lightgray] (F) -- (G);
      \draw[line width=2pt] (F) -- (H);
      \draw[line width=2pt] (G) -- (H);

      \node[p0,left of=G] (I) {};
      \node[p0, left of =H] (J) {};
      \draw[lightgray] (G) -- (I);  
      \draw[line width=2pt] (H) -- (I);
      \draw[lightgray] (H) -- (J);      
      \draw[line width=2pt] (J) -- (I);
      \draw[out=135,in=45, loop, looseness=6,line width=2pt] (I) edge  (I) {};

      \node[p0] at (-4.5,3.5) (K) {};
      \draw[arrows = {->[harpoon,swap]},blue,line width=2pt] ($(A)!0.5!(K)$) -- (K);
      \draw[arrows = {->[harpoon,swap]},lightgray] ($(K)!0.5!(A)$) -- node [] {} (A);

      \draw[arrows = {->[harpoon,swap]},blue,line width=2pt] ($(C)!0.5!(K)$) -- (K);
      \draw[arrows = {->[harpoon,swap]},lightgray] ($(K)!0.5!(C)$) -- node [] {} (C);

      \draw[arrows = {->[harpoon]},lightgray] ($(E)!0.5!(K)$) -- (K);
      \draw[arrows = {->[harpoon]},lightgray] ($(K)!0.5!(E)$) -- node [] {} (E);

      \draw[arrows = {->[harpoon]},lightgray] ($(G)!0.5!(K)$) -- (K);
      \draw[arrows = {->[harpoon]},blue,line width=2pt] ($(K)!0.5!(G)$) -- node [] {} (G);

    \end{tikzpicture}%
    }
    \caption{The canonical policy~$\pi_8$ for~$\mathcal{B}_4$, which results from applying all canonical phases w.r.t.\ 7 to~$\pi_7$.}
    \label{fig:c8}
  \end{subfigure}
 \caption{An example that illustrates how the canonical phases transform one canonical policy into the next one. Active edges are depicted in a bold blue color, while inactive edges are slightly transparent. Note that~$\pi_7$ enters the first three levels, which correspond to the set bits in the binary representation of~$7$; analogously,~$\pi_8$ only enters the fourth level.}
    \label{fig:c7toc8}
\end{figure}

The following Lemma shows that if the algorithm applies the canonical phases to the corresponding canonical policy, it reaches the next canonical policy.
Consider Figure~\ref{fig:c7toc8} for an example.

\begin{restatable}{lemma}{lemcanphases}\label{odd phases result in next canonical}
    Let~$\pi_x$ be the canonical policy for some odd~$x\in[2^n-3]$. 
    Applying the canonical phases with respect to~$x$ to~$\pi_x$ results in the canonical policy~$\pi_{x+1}$.
\end{restatable}
\begin{proof}
    Let~$x\in[2^n-3]$ be odd, and let~$\pi_x$ be the canonical policy for~$x$. 
    For convenience, we write~$\ell\coloneqq\ell_0(x)>1$.
    Let~$\pi$ denote the policy resulting from applying the canonical phases w.r.t.\ $x$ to~$\pi_x$.
    We need to show that~$\pi$ satisfies the properties from Definition~\ref{canonical policy} with respect to~$x+1$.
    
    Due to the third canonical phase,~$\travel(\ell)$ is active in~$\pi$.
    Since~$\ell_1(x+1)=\ell$, this implies that~$\pi$ satisfies condition <1> w.r.t~$x+1$.
    
    We have~$\m(x+1)\geq \m(x)$, and the canonical phases solely include switches in the first~$\ell$ levels.
    Therefore, if~$\ell<\m(x)$, level~$\m(x+1)$ and the levels above remain unchanged.     
    If~$\ell\geq \m(x)$, we know that~$\ell=\m(x)+1=\m(x+1)$ and~$x_{\ell+1}=0$. Hence, the only switches in level~$\m(x+1)$ or above are~$\enter(\ell)$ and~$\travel(\ell)$.
    In both cases, we obtain that, as~$\pi_x$ satisfies property <2> w.r.t.\ $x$, policy~$\pi$ satisfies property <2> w.r.t.\ $x+1$.
    
    Note that the bit configurations of~$x+1$ and~$x$ differ precisely in the first~$\ell>1$ bits, where we have~$(x+1)_{\ell}=1$ and~$(x+1)_j=0$ for all~$j\in[\ell-1]$. 
    Due to the third canonical phase,~$\enter(\ell)$ is active in~$\pi$; and due to the fourth and fifth phases,~$\enter(i)$ is inactive for all~$i\in[\ell-1]$. 
    Further, the canonical phases do not contain switches in levels above level~$\ell$.
    Since~$\pi_x$ satisfies <a> w.r.t.\ $x$, we can conclude that policy~$\pi$ satisfies <a> w.r.t.\ $x+1$.
    
    If~$(x+1)_2=1$, then~$x$ being odd yields~$\ell=2$.
    Thus, the fifth and seventh phases ensure that~$\skipp(1)$ and~$\leave(1)$ are active in~$\pi$.
    This yields that~$\pi$ has property~<b>.
    
    Now consider some~$i\geq 3$ with~$(x+1)_i=(x+1)_{i-1}=1$. Since~$(x+1)_j=0$ for all~$j\in[\ell-1]$, this yields~$i-1\geq\ell$. 
    
    If~$i-1=\ell$, then~$x_{\ell+1}=(x+1)_{\ell+1}=(x+1)_i=1$, so the first phase ensures that~$\leave(i-1)$ is active in~$\pi$.    
    Otherwise, none of the phases include a switch in level~$i-1$ or above, and we have~$x_i=(x+1)_i=1$ as well as~$x_{i-1}=(x+1)_{i-1}=1$. 
    Thus, property <c> w.r.t.\ $x$ yields that~$\leave(i-1)$ is active in~$\pi_x$. 
    Therefore,~$\leave(i-1)$ is still active in~$\pi$.
    We conclude that~$\pi$ satisfies property <c>.
    
    Now consider some~$i\geq3$ with~$(x+1)_i=1$ and~$(x+1)_{i-1}=0$.
    To show that~<$\text{d}_1$> and~<$\text{d}_2$> are satisfied, we consider the cases~$\ell<i$ and~$\ell\geq i$.
    
    First, assume that~$\ell<i$. 
    Then,~$(x+1)_\ell=1$ yields that~$\ell\neq i-1$, so~$\ell\leq i-2$. 
    Therefore, we have~$x_i=(x+1)_i=1$ and~$x_{i-1}=(x+1)_{i-1}=0$. 
    As~$\pi_x$ satisfies property~<$\text{d}_1$> w.r.t.\ $x$ and as the phases do not contain switches above level~$\ell$, we obtain that~$\stay(i-1)$ and~$\skipp(i-1)$ are still active in~$\pi$. 
    Further, the only switches in level~$\ell$ are~$\leave(\ell)$,~$\enter(\ell)$, and~$\travel(\ell)$. 
    Therefore,~$\leave(i-2)$ is also still active in~$\pi$. 
    Hence, in case~$\ell<i$, property~<$\text{d}_1$> is satisfied.
    
    Due to~$(x+1)_{\ell}=1$ and~$\ell\leq i-2$, we have~$\Lx i {x+1}\geq\ell>1$. 
    As condition~<$\text{d}_2$> \mbox{w.r.t.\ $x+1$} becomes trivial otherwise, assume~$\Lx i {x+1}<i-2$ in the following.

    Since the bit configurations of~$x$ and~$x+1$ do not differ from each other above bit~$\ell$ and~$\Lx i {x+1}\geq\ell$, we have~$\Lx i {x+1}\geq \Lx i {x}$.
    policy~$\pi_x$ satisfies~<$\text{d}_2$> \mbox{w.r.t.\ $x$} and the canonical phases do not apply switches above level~$\ell$.
    Therefore, under the condition that~$\stay(\ell)$ is active in~$\pi$ if~$\Lx i {x+1}=\ell$, we obtain that~$\pi$ satisfies~<$\text{d}_2$> w.r.t.\ $x+1$.
    
    To verify this condition, assume~$\Lx i {x+1}=\ell$. 
    Then,~$\Lx ix\leq\ell<i-2$ and property~<$\text{d}_2$> yields that~$\stay(\ell)$ is active in~$\pi_x$.
    The switch~$\leave(\ell)$ from the first phase is the only switch which can cause that~$\stay(\ell)$ is not active in~$\pi$ anymore. 
    However, if~$x_{\ell+1}=1$, we also have~$(x+1)_{\ell+1}=1$.
    This yields~$\Lx i {x+1}\geq \ell+1$, which contradicts~$\Lx i {x+1}=\ell$. 
    Hence,~$\stay(\ell)$ is still active in~$\pi$, which verifies our condition. 
    We conclude that, in case~$\ell<i$, policy~$\pi$ satisfies property~<$\text{d}_2$>.
    
    Second, we assume~$\ell\geq i$. 
    As~$(x+1)_i=1$ and~$(x+1)_j=0$ for all~$j\in[\ell-1]$, we then have~$\ell=i$ and~$\Lx i {x+1}=1$.  
    
    Due to the fifth phase,~$\skipp(\ell-1)$ is active in~$\pi$. 
    Further,~$x_{\ell-1}=x_{\ell-2}=1$ and properties~<b> and <c> yield that~$\leave(\ell-2)$ is active in~$\pi_x$.
    Since none of the phases includes the switch~$\stay(\ell-2)$, we can conclude that~$\leave(\ell-2)$ is still active in~$\pi$.
    For property~<$\text{d}_1$>, it remains to show that~$\stay(\ell-1)$ is active in~$\pi$. 
    
    We can assume that~$x_{\ell+1}=0$ and~$\ell\leq\m(x)$ as otherwise the second phase yields that~$\stay(\ell-1)$ is active in~$\pi$. 
    Then, there is some~$j>\ell+1$ with~$\Lx jx=\ell-1<j-2$ and~$x_j=1$.
    Hence, property~<$\text{d}_2$> w.r.t.\ $x$ yields that~$\stay(\ell-1)$ is active \mbox{in~$\pi_x$}. 
    As we have~$\ell=i\geq3$, none of the phases includes the switch~$\leave(\ell-1)$, so~$\stay(\ell-1)$ is still active in~$\pi$.   
    Hence, policy~$\pi$ satisfies property~<$\text{d}_1$>. 
    
    We have~$\ell=i\geq3$, so~$\board(1)$ is active in~$\pi$ due to the fourth phase.
    Since it holds that~$\Lx i {x+1}=1$, we obtain that~$\pi$ satisfies property~<$\text{d}_2$> w.r.t.\ $x+1$ if~$i=3$. 
    We thus assume~$\ell\geq4$.
    Then, the fourth phase yields that~$\board(j)$ is active for all~$j\in[i-2]$, and the sixth phase yields that~$\stay(j)$ is active for all~$j\in[i-3]$.
    Therefore, policy~$\pi$ satisfies condition~<$\text{d}_2$> w.r.t.\ $x+1$. 
    This concludes the proof.
\end{proof}

Finally, we show that \textsc{Bland} actually applies the canonical phases when given the corresponding canonical policy. 

\begin{restatable}{lemma}{lemoddcanphases}\label{Bland applies odd can phases}
    Let~$x\in[2^n-3]$ be odd.
    Then,~$\textsc{Bland}(\mbl,\pi_x,\order)$ visits~$\pi_{x+1}$. 
\end{restatable}
\begin{proof}
    Let~$x\in[2^n-3]$ be odd, let~$\pi_x$ be the canonical policy for~$x$, and write~$\ell\coloneqq\ell_0(x)\geq2$.
    By Lemma~\ref{odd phases result in next canonical}, it suffices to show that \textsc{Bland} applies the canonical phases w.r.t.\ $x$ to~$\pi_x$.
    
    By properties~<a>,~<b>, and~<c> from Definition~\ref{canonical policy}, the edges~$\enter(\ell-1)$,~$\enter(i)$ and~$\leave(i)$ are active in~$\pi_x$ for all~$i\in[\ell-2]$.
    Further,~$\travel(1)$ is active, so there are clearly no improving switches for~$\pi_x$ in the first~$\ell-2$ levels. 
    By Lemma \ref{no pointer improving}, there are also no improving switches for~$\pi_x$ in~$t$.

    The first two canonical phases depend on the structure of~$x$, so we consider all possible cases here.
    First, assume~$x_{\ell+1}=1$. 
    Then, by properties~<a> and~<$\text{d}_1$>,~$\enter(\ell+1)$, $\stay(\ell)$,~$\skipp(\ell)$, and~$\leave(\ell-1)$ are active in~$\pi_x$, cf. Figure~\ref{fig:for Bland proof a}.
    Hence, there is no improving switch in level~$\ell-1$, and~$\leave(\ell)$ is the only improving switch in level~$\ell$.
    Since the Bland numbering~$\order$ prefers low levels, the algorithm applies~$\leave(\ell)$ to~$\pi_x$.

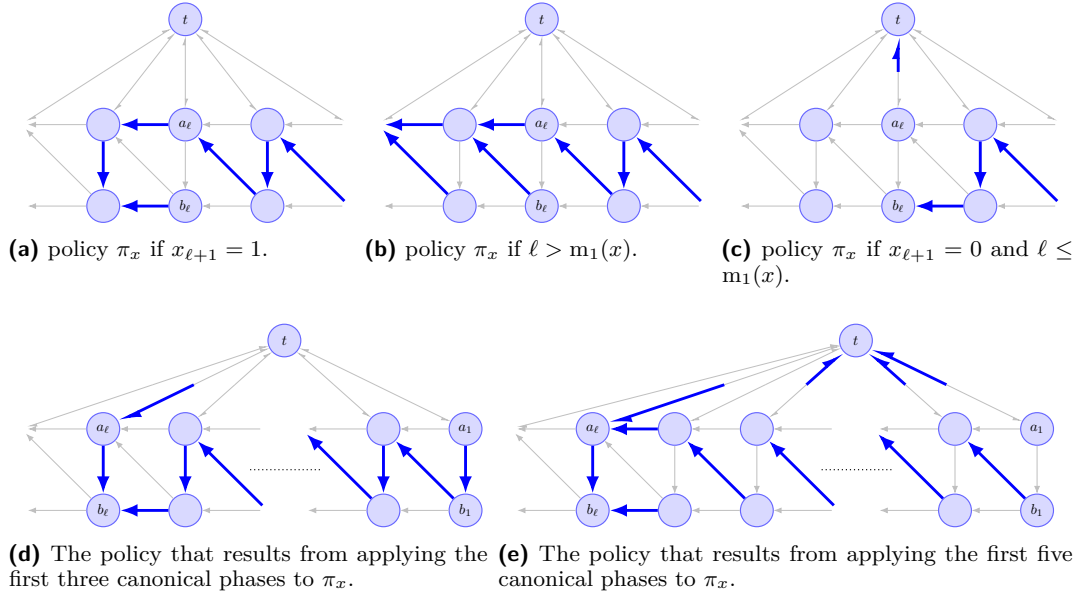
\begin{figure}
\centering
\begin{subfigure}{0.33\textwidth}
\centering
        \scalebox{0.54}{
        \begin{tikzpicture}[node distance = 2cm,
      p0/.style={circle, draw=blue!60, fill=blue!15, thick, minimum size=8mm},
      inv/.style={circle, draw=white, fill=white},
      bl/.style={fill=white,draw=gray!30,text=gray,font=\footnotesize},
      every path/.style={->,draw=blue,line width=.5pt,{>=Latex}}]
    
      \node[inv] (A) {};
      \node[inv,below of=A] (B) {};
      
      \node[p0,left of=A] (C) {};
      \node[p0,left of=B] (D) {};
      \draw[lightgray] (A) -- (C);
      \draw[lightgray] (B) -- (D);
      \draw[line width=2pt] (B) -- (C);
      \draw[line width=2pt] (C) -- (D);
      
      \node[p0,left of=C] (E) {$a_{\ell}$};
      \node[p0,left of=D] (F) {$b_{\ell}$};
      \draw[lightgray] (C) -- (E);
      \draw[line width=2pt] (D) -- (E);
      \draw[lightgray] (D) -- (F);
      \draw[lightgray] (E) -- (F);

      \node[p0,left of=E] (G) {};
      \node[p0,left of=F] (H) {};
      \draw[line width=2pt] (E) -- (G);      
      \draw[lightgray] (F) -- (G);
      \draw[line width=2pt] (F) -- (H);
      \draw[line width=2pt] (G) -- (H);

      \node[inv,left of=G] (I) {};
      \node[inv,left of=H] (J) {};
      \draw[lightgray] (G) -- (I);      
      \draw[lightgray] (H) -- (I);
      \draw[lightgray] (H) -- (J);
      
      \node[p0, above = 5em of E] (O) {$t$};
      \draw[arrows = {->[harpoon,swap]},lightgray] ($(A)!0.5!(O)$) -- (O);
      \draw[arrows = {->[harpoon,swap]},lightgray] ($(O)!0.5!(A)$) -- node [] {} (A);

      \draw[arrows = {->[harpoon,swap]},lightgray] ($(C)!0.5!(O)$) -- (O);
      \draw[arrows = {->[harpoon,swap]},lightgray] ($(O)!0.5!(C)$) -- node [] {} (C);

      \draw[arrows = {->[harpoon]},lightgray] ($(E)!0.5!(O)$) -- (O);
      \draw[arrows = {->[harpoon]},lightgray] ($(O)!0.5!(E)$) -- node [] {} (E);

      \draw[arrows = {->[harpoon]},lightgray] ($(G)!0.5!(O)$) -- (O);
      \draw[arrows = {->[harpoon]},lightgray] ($(O)!0.5!(G)$) -- node [] {} (G);
      
      \draw[arrows = {->[harpoon]},lightgray] ($(I)!0.5!(O)$) -- (O);
      \draw[arrows = {->[harpoon]},lightgray] ($(O)!0.5!(I)$) -- node [] {} (I);
    
    \end{tikzpicture}
    }
    
    \caption{policy~$\pi_x$ if~$x_{\ell+1}=1$.\\~}\label{fig:for Bland proof a}
    \end{subfigure}\hfill
    \begin{subfigure}{0.33\textwidth}
        \scalebox{0.54}{
        \begin{tikzpicture}[node distance = 2cm,
      p0/.style={circle, draw=blue!60, fill=blue!15, thick, minimum size=8mm},
      inv/.style={circle, draw=white, fill=white},
      bl/.style={fill=white,draw=gray!30,text=gray,font=\footnotesize},
      every path/.style={->,draw=blue,line width=.5pt,{>=Latex}}]
    
      \node[inv] (A) {};
      \node[inv,below of=A] (B) {};
      
      \node[p0,left of=A] (C) {};
      \node[p0,left of=B] (D) {};
      \draw[lightgray] (A) -- (C);
      \draw[lightgray] (B) -- (D);
      \draw[line width=2pt] (B) -- (C);
      \draw[line width=2pt] (C) -- (D);
      
      \node[p0,left of=C] (E) {$a_{\ell}$};
      \node[p0,left of=D] (F) {$b_{\ell}$};
      \draw[lightgray] (C) -- (E);
      \draw[line width=2pt] (D) -- (E);
      \draw[lightgray] (D) -- (F);
      \draw[lightgray] (E) -- (F);

      \node[p0,left of=E] (G) {};
      \node[p0,left of=F] (H) {};
      \draw[line width=2pt] (E) -- (G);      
      \draw[line width=2pt] (F) -- (G);
      \draw[lightgray] (F) -- (H);
      \draw[lightgray] (G) -- (H);

      \node[inv,left of=G] (I) {};
      \node[inv,left of=H] (J) {};
      \draw[line width=2pt] (G) -- (I);      
      \draw[line width=2pt] (H) -- (I);
      \draw[lightgray] (H) -- (J);
      
      \node[p0, above = 5em of E] (O) {$t$};
      \draw[arrows = {->[harpoon,swap]},lightgray] ($(A)!0.5!(O)$) -- (O);
      \draw[arrows = {->[harpoon,swap]},lightgray] ($(O)!0.5!(A)$) -- node [] {} (A);

      \draw[arrows = {->[harpoon,swap]},lightgray] ($(C)!0.5!(O)$) -- (O);
      \draw[arrows = {->[harpoon,swap]},lightgray] ($(O)!0.5!(C)$) -- node [] {} (C);

      \draw[arrows = {->[harpoon]},lightgray] ($(E)!0.5!(O)$) -- (O);
      \draw[arrows = {->[harpoon]},lightgray] ($(O)!0.5!(E)$) -- node [] {} (E);

      \draw[arrows = {->[harpoon]},lightgray] ($(G)!0.5!(O)$) -- (O);
      \draw[arrows = {->[harpoon]},lightgray] ($(O)!0.5!(G)$) -- node [] {} (G);
      
      \draw[arrows = {->[harpoon]},lightgray] ($(I)!0.5!(O)$) -- (O);
      \draw[arrows = {->[harpoon]},lightgray] ($(O)!0.5!(I)$) -- node [] {} (I);
    
    \end{tikzpicture}
    }
    
    \caption{policy~$\pi_x$ if~$\ell>\m(x)$.\\~}\label{fig:for Bland proof b}
    \end{subfigure}\hfill
    \begin{subfigure}{0.33\textwidth}
        \scalebox{0.54}{
        \begin{tikzpicture}[node distance = 2cm,
      p0/.style={circle, draw=blue!60, fill=blue!15, thick, minimum size=8mm},
      inv/.style={circle, draw=white, fill=white},
      bl/.style={fill=white,draw=gray!30,text=gray,font=\footnotesize},
      every path/.style={->,draw=blue,line width=.5pt,{>=Latex}}]
    
      \node[inv] (A) {};
      \node[inv,below of=A] (B) {};
      
      \node[p0,left of=A] (C) {};
      \node[p0,left of=B] (D) {};
      \draw[lightgray] (A) -- (C);
      \draw[lightgray] (B) -- (D);
      \draw[line width=2pt] (B) -- (C);
      \draw[line width=2pt] (C) -- (D);
      
      \node[p0,left of=C] (E) {$a_{\ell}$};
      \node[p0,left of=D] (F) {$b_{\ell}$};
      \draw[lightgray] (C) -- (E);
      \draw[lightgray] (D) -- (E);
      \draw[line width=2pt] (D) -- (F);
      \draw[lightgray] (E) -- (F);

      \node[p0,left of=E] (G) {};
      \node[p0,left of=F] (H) {};
      \draw[lightgray] (E) -- (G);      
      \draw[lightgray] (F) -- (G);
      \draw[lightgray] (F) -- (H);
      \draw[lightgray] (G) -- (H);

      \node[inv,left of=G] (I) {};
      \node[inv,left of=H] (J) {};
      \draw[lightgray] (G) -- (I);      
      \draw[lightgray] (H) -- (I);
      \draw[lightgray] (H) -- (J);

      \node[p0, above = 5em of E] (O) {$t$};
      \draw[arrows = {->[harpoon,swap]},lightgray] ($(A)!0.5!(O)$) -- (O);
      \draw[arrows = {->[harpoon,swap]},lightgray] ($(O)!0.5!(A)$) -- node [] {} (A);

      \draw[arrows = {->[harpoon,swap]},lightgray] ($(C)!0.5!(O)$) -- (O);
      \draw[arrows = {->[harpoon,swap]},lightgray] ($(O)!0.5!(C)$) -- node [] {} (C);

      \draw[arrows = {->[harpoon]},blue,line width=2pt] ($(E)!0.5!(O)$) -- (O);
      \draw[arrows = {->[harpoon]},lightgray] ($(O)!0.5!(E)$) -- node [] {} (E);

      \draw[arrows = {->[harpoon]},lightgray] ($(G)!0.5!(O)$) -- (O);
      \draw[arrows = {->[harpoon]},lightgray] ($(O)!0.5!(G)$) -- node [] {} (G);
      
      \draw[arrows = {->[harpoon]},lightgray] ($(I)!0.5!(O)$) -- (O);
      \draw[arrows = {->[harpoon]},lightgray] ($(O)!0.5!(I)$) -- node [] {} (I);
    
    \end{tikzpicture}
        }
    
    \caption{policy~$\pi_x$ if~$x_{\ell+1}=0$ and~$\ell\leq\m(x)$.}\label{fig:for Bland proof c}
    \end{subfigure}
        
        \bigskip
        
        \begin{subfigure}{0.45\textwidth}
        \scalebox{0.54}{
        \begin{tikzpicture}[node distance = 2cm,
      p0/.style={circle, draw=blue!60, fill=blue!15, thick, minimum size=8mm},
      inv/.style={circle, draw=white, fill=white},
      bl/.style={fill=white,draw=gray!30,text=gray,font=\footnotesize},
      every path/.style={->,draw=blue,line width=.5pt,{>=Latex}}]
    
      \node[p0] (A) {$a_1$};
      \node[p0,below of=A] (B) {$b_1$};
      \draw[line width=2pt] (A) -- (B);
      
      \node[p0,left of=A] (C) {};
      \node[p0,left of=B] (D) {};
      \draw[lightgray] (A) -- (C);
      \draw[lightgray] (B) -- (D);
      \draw[line width=2pt] (B) -- (C);
      \draw[line width=2pt] (C) -- (D);

      \node[inv,left of=C] (C') {};
      \node[inv,left of=D] (D') {};
      \draw[lightgray] (C) -- (C');
      \draw[lightgray] (D) -- (D');
      \draw[line width=2pt] (D) -- (C');
      
      \node[p0,left = 4cm of C] (E) {};
      \node[p0,left = 4cm of D] (F) {};
      \node[inv, left = 1.2cm of $(C)!0.5!(D)$] (help1) {};
      \node[inv, right = 1.2cm of $(E)!0.5!(F)$] (help2) {};
      \draw[thick, dotted,black,-] (help1) -- (help2);
      \draw[line width=2pt] (E) -- (F);

      \node[inv,right of=E] (E') {};
      \node[inv,right of=F] (F') {};
      \draw[lightgray] (E') -- (E);
      \draw[lightgray] (F') -- (F);
      \draw[line width=2pt] (F') -- (E);

      \node[p0,left of=E] (G) {$a_{\ell}$};
      \node[p0,left of=F] (H) {$b_{\ell}$};
      \draw[lightgray] (E) -- (G);      
      \draw[lightgray] (F) -- (G);
      \draw[line width=2pt] (F) -- (H);
      \draw[line width=2pt] (G) -- (H);

      \node[inv,left of=G] (I) {};
      \node[inv,left of=H] (J) {};
      \draw[lightgray] (G) -- (I);      
      \draw[lightgray] (H) -- (I);
      \draw[lightgray] (H) -- (J);

      \node[p0, above = 5em of $(C)!0.5!(E)$] (O) {$t$};
      \draw[arrows = {->[harpoon,swap]},lightgray] ($(A)!0.5!(O)$) -- (O);
      \draw[arrows = {->[harpoon,swap]},lightgray] ($(O)!0.5!(A)$) -- node [] {} (A);

      \draw[arrows = {->[harpoon,swap]},lightgray] ($(C)!0.5!(O)$) -- (O);
      \draw[arrows = {->[harpoon,swap]},lightgray] ($(O)!0.5!(C)$) -- node [] {} (C);

      \draw[arrows = {->[harpoon]},lightgray] ($(E)!0.5!(O)$) -- (O);
      \draw[arrows = {->[harpoon]},lightgray] ($(O)!0.5!(E)$) -- node [] {} (E);

      \draw[arrows = {->[harpoon]},lightgray] ($(G)!0.5!(O)$) -- (O);
      \draw[arrows = {->[harpoon]},blue, line width=2pt] ($(O)!0.5!(G)$) -- node [] {} (G);
      
      \draw[arrows = {->[harpoon]},lightgray] ($(I)!0.5!(O)$) -- (O);
      \draw[arrows = {->[harpoon]},lightgray] ($(O)!0.5!(I)$) -- node [] {} (I);
    
    \end{tikzpicture}
    }
    
    \caption{The policy that results from applying the first three canonical phases to~$\pi_x$.}\label{fig:for Bland proof d}
    \end{subfigure}\hfill
    \begin{subfigure}{0.54\textwidth}
        \scalebox{0.54}{
        \begin{tikzpicture}[node distance = 2cm,
      p0/.style={circle, draw=blue!60, fill=blue!15, thick, minimum size=8mm},
      inv/.style={circle, draw=white, fill=white},
      bl/.style={fill=white,draw=gray!30,text=gray,font=\footnotesize},
      every path/.style={->,draw=blue,line width=.5pt,{>=Latex}}]
    
      \node[p0] (A) {$a_1$};
      \node[p0,below of=A] (B) {$b_1$};
      \draw[lightgray] (A) -- (B);
      
      \node[p0,left of=A] (C) {};
      \node[p0,left of=B] (D) {};
      \draw[lightgray] (A) -- (C);
      \draw[lightgray] (B) -- (D);
      \draw[line width=2pt] (B) -- (C);
      \draw[lightgray] (C) -- (D);

      \node[inv,left of=C] (C') {};
      \node[inv,left of=D] (D') {};
      \draw[lightgray] (C) -- (C');
      \draw[lightgray] (D) -- (D');
      \draw[line width=2pt] (D) -- (C');
      
      \node[p0,left = 4cm of C] (E) {};
      \node[p0,left = 4cm of D] (F) {};
      \node[inv, left = 1.2cm of $(C)!0.5!(D)$] (help1) {};
      \node[inv, right = 1.2cm of $(E)!0.5!(F)$] (help2) {};
      \draw[thick, dotted,black,-] (help1) -- (help2);
      \draw[lightgray] (E) -- (F);

      \node[inv,right of=E] (E') {};
      \node[inv,right of=F] (F') {};
      \draw[lightgray] (E') -- (E);
      \draw[lightgray] (F') -- (F);
      \draw[line width=2pt] (F') -- (E);

      \node[p0,left of=E] (E'') {};
      \node[p0,left of=F] (F'') {};
      \draw[lightgray] (E) -- (E'');
      \draw[lightgray] (F) -- (F'');
      \draw[line width=2pt] (F) -- (E'');
      \draw[lightgray] (E'') -- (F'');

      \node[p0,left of=E''] (G) {$a_{\ell}$};
      \node[p0,left of=F''] (H) {$b_{\ell}$};
      \draw[line width=2pt] (E'') -- (G);      
      \draw[lightgray] (F'') -- (G);
      \draw[line width=2pt] (F'') -- (H);
      \draw[line width=2pt] (G) -- (H);

      \node[inv,left of=G] (I) {};
      \node[inv,left of=H] (J) {};
      \draw[lightgray] (G) -- (I);      
      \draw[lightgray] (H) -- (I);
      \draw[lightgray] (H) -- (J);

      \node[p0, above = 5em of $(C)!0.5!(E)$] (O) {$t$};
      \draw[arrows = {->[harpoon,swap]},blue, line width=2pt] ($(A)!0.5!(O)$) -- (O);
      \draw[arrows = {->[harpoon,swap]},lightgray] ($(O)!0.5!(A)$) -- node [] {} (A);

      \draw[arrows = {->[harpoon,swap]},blue, line width=2pt] ($(C)!0.5!(O)$) -- (O);
      \draw[arrows = {->[harpoon,swap]},lightgray] ($(O)!0.5!(C)$) -- node [] {} (C);

      \draw[arrows = {->[harpoon]},blue, line width=2pt] ($(E)!0.5!(O)$) -- (O);
      \draw[arrows = {->[harpoon]},lightgray] ($(O)!0.5!(E)$) -- node [] {} (E);
      
      \draw[arrows = {->[harpoon]},lightgray] ($(E'')!0.5!(O)$) -- (O);
      \draw[arrows = {->[harpoon]},lightgray] ($(O)!0.5!(E'')$) -- node [] {} (E'');

      \draw[arrows = {->[harpoon]},lightgray] ($(G)!0.5!(O)$) -- (O);
      \draw[arrows = {->[harpoon]},blue, line width=2pt] ($(O)!0.5!(G)$) -- node [] {} (G);
      
      \draw[arrows = {->[harpoon]},lightgray] ($(I)!0.5!(O)$) -- (O);
      \draw[arrows = {->[harpoon]},lightgray] ($(O)!0.5!(I)$) -- node [] {} (I);
    
    \end{tikzpicture}
        }
    
    \caption{The policy that results from applying the first five canonical phases to~$\pi_x$.}\label{fig:for Bland proof e}
    \end{subfigure}
        \caption{A collection of figures supporting the proof of Lemma~\ref{Bland applies odd can phases}, where~$\pi_x$ denotes the canonical policy for some odd~$x\in[2^n-3]$. Bold edges denote active edges in the respective policy; some of the remaining edges might be active as well.}
        \label{fig:for Bland proof}
    \end{figure}

    After this switch, the edges~$\enter(\ell)$ and~$\stay(\ell-1)$ are the only improving switches in the first~$\ell$ levels.
    Thus,~$\stay(\ell-1)$ gets applied next.

    Second, assume~$\ell>\m(x)$. This yields~$\ell=\m(x)+1$ and~$x_{\ell+1}=0$.
    Due to property~<2>, the edges~$\leave(\ell-1)$,~$\skipp(i)$, and~$\leave(i)$ are active for all~$i\geq\ell$, cf. Figure~\ref{fig:for Bland proof b}.
    Hence, the only improving switch in the first~$\ell-1$ levels is~$\stay(\ell-1)$, which gets applied next.

    Lastly, assume~$x_{\ell+1}=0$ and~$\ell\leq\m(x)$. 
    Then, there is some~$i>\ell+1$ with~$x_i=1$ and~$\Lx ix=\ell-1<i-2$.
    Hence, property~<$\text{d}_2$> yields that~$\board(\ell)$ and~$\stay(\ell-1)$ are active in~$\pi_x$, cf. Figure~\ref{fig:for Bland proof c}.
    
    Let~$\pi$ be the policy resulting from the application of the first and second canonical phases to~$\pi_x$.
    From the case distinction above, we can conclude that~$\pi$ includes the path~$(t,a_1,b_1,a_2,b_2,\ldots,a_{\ell-1},b_{\ell-1},b_\ell)$.
    Further, we either know that~$\leave(\ell)$ and~$\skipp(\ell)$ are active or we know that~$\board(\ell)$ is active.
    In both cases, there are no improving switches in the first~$\ell-1$ levels and~$\enter(\ell)$ is improving.
    Hence, by definition of~$\order$,~$\enter(\ell)$ gets applied next.

    Entering level~$\ell$ yields a higher reward than entering all of the first~$\ell-1$ levels, so the edge~$\travel(\ell)$ is improving now.
    Since the~$\ell-1$ travel edges with a smaller Bland number are not improving, the algorithm applies~$\travel(\ell)$ next.

    Note that, by Lemma \ref{no pointer improving}, travel edges are not improving for~$\pi_x$, and since travel edges above level~$\ell$ have certainly not become improving during the application of the first three phases, these edges are still not improving for the current policy.    
    Consider Figure~\ref{fig:for Bland proof d} for the structure of the current policy on the first~$\ell$ levels.

    The first~$\ell-1$ levels contain the improving switches~$\skipp(\ell-1)$,~$\leave(\ell-1)$, and~$\board(j)$ for~$j\in[\ell-1]$ as these edges allow vertices of lower levels to enter level~$\ell$.
    The Bland numbering~$\order$ yields that \textsc{Bland} applies~$\board(1)$ if~$\ell\geq3$.
    This switch does not create new improving switches.
    Following this argument, we obtain that the algorithm applies~$\board(j)$ for all~$j\in[\ell-2]$ in increasing order.
    Since~$\skipp(\ell-1)$ is still improving after that and since it precedes~$\board(\ell-1)$ and~$\leave(\ell-1)$ in the Bland numbering, it gets applied next.

    Let~$\pi'$ denote the current policy, i.e., the policy resulting from the application of the first five phases to~$\pi_x$. 
    The structure of~$\pi'$ on the first~$\ell$ levels is given in Figure~\ref{fig:for Bland proof e}. 
    Note that the travel edges still have not become improving during the latest switches.

    If~$\ell=2$, then~$\leave(1)$ is the only improving switch for~$\pi'$ in the first level.
    Hence, the algorithm applies this switch and the canonical phases are completed.

    They are also completed if~$\ell=3$, so assume~$\ell\geq4$ in the following.    
    Applying~$\skipp(i)$ or~$\enter(i)$ in some level~$i\in[\ell-3]$ is not improving for~$\pi'$ since the costs of~$\travel(i+1)$ are higher than the costs of~$\travel(i)$.
    Similarly,~$\stay(i)$ is not improving for~$i\in[\ell-4]$ as~$\travel(i+2)$ is more expensive than~$\travel(i+1)$.
    
    However, by switching to~$\stay(\ell-3)$, the agent avoids to use the edge~$\travel(\ell-2)$ and reaches~$a_\ell$ via~$\leave(\ell-2)$ and~$\skipp(\ell-1)$ instead.
    Hence,~$\stay(\ell-3)$ is the only improving switch in the first~$\ell-3$ levels and gets applied next.

    This switch makes~$\enter(\ell-3)$ and~$\stay(\ell-4)$ improving, so the Bland numbering~$\order$ yields that the algorithm applies~$\stay(\ell-4)$ next.
    By iterating this argument, we obtain that \textsc{Bland} applies the improving switch~$\stay(j)$ for all~$j\in[\ell-3]$ in decreasing order, which completes the canonical phases entirely.

\end{proof}

According to Lemma \ref{Bland applies even can phases}, \textsc{Bland} transforms every even canonical policy~$\pi_x$ into~$\pi_{x+1}$, and by Lemma \ref{Bland applies odd can phases}, the same holds for odd canonical policies.
Since the initial policy is canonical for zero, this yields that~$\textsc{Bland}(\mbl,\pi_0,\order)$ visits all canonical \mbox{policies}~$\pi_i$ with~$i\in[2^n-1]$.
Since these are pairwise different, this proves the main result of this section.

\begin{theorem}\label{exp auf B}
   There is an initial policy such that the policy iteration algorithm with Bland's pivot rule performs~$\Omega(2^n)$ improving switches when applied to~$\mbl$.
\end{theorem}

We close this section with two technical observations that help us later.

\begin{restatable}{observation}{obsenterhighestrc}\label{red costs in ai}
    For \mbox{every~$i\in[n]$}, whenever~$\textsc{Bland}(\mbl,\pi_0,\order)$ applies the improving switch~$\skipp(i)$, this edge has higher reduced costs than~$\board(i)$; whenever it applies~$\enter(i)$, this edge has higher reduced costs than~$\skipp(i)$ and~$\board(i)$.
\end{restatable}
\begin{proof}
    Assume that~$\skipp(i)$ gets applied by~$\textsc{Bland}(\mbl,\pi_0,\order)$ and denote the resulting policy by~$\pi$.
    We know that, at this point, the algorithm completed the fifth canonical phase w.r.t.\ some odd~$x\in[2^n-3]$, and~$i=\ell_0(x)-1\eqqcolon\ell-1$.
    In the proof of Lemma~\ref{Bland applies odd can phases}, we argued that Figure~\ref{fig:for Bland proof e} shows the structure of~$\pi$ on the first~$\ell$ levels.
    We see that~$\travel(\ell)$ is active in~$\pi$, which immediately yields that the reduced costs of the applied switch~$\skipp(\ell-1)$ were higher than the ones of~$\board(\ell-1)$.

    Now assume that the algorithm applies the improving switch~$\enter(i)$ to the current policy~$\pi$.
    First, assume that this happens during the third canonical phase w.r.t.\ some odd~$x\in[2^n-3]$.
    Then~$i=\ell_0(x)\eqqcolon\ell$.
    Further, we argue in the proof of Lemma~\ref{Bland applies odd can phases} that~$\pi$ includes the path~$P=(t,a_1,b_1,a_2,b_2,\ldots,a_{\ell-1},b_{\ell-1},b_\ell)$ and that either~$\skipp(\ell)$ or~$\board(\ell)$ is active in~$\pi$.
    
    If~$\skipp(\ell)$ is active, the reduced costs of~$\enter(\ell)$ exceed the ones of~$\board(\ell)$ as entering level~$\ell$ yields a higher reward than following path~$P$.
    Otherwise,~$\board(\ell)$ is active and condition~<d> from Definition~\ref{canonical policy} yields that either~$\board(\ell+1)$ or~$\leave(\ell)$ is still active in~$\pi$ as it was active in~$\pi_x$.
    In the first case,~$\skipp(\ell)$ is not improving for~$\pi$ since boarding~$t$ from a higher level is more expensive.
    In the other case, entering (and leaving) level~$\ell$ is more improving than skipping this level.

    Now assume that the switch~$\enter(i)$ is not part of the third canonical phase.
    Then, according to the proof of Lemma~\ref{Bland applies even can phases}, it gets applied to the canonical policy~$\pi_x$ for some even~$x\in[2^n-2]_0$ and we have~$i=1$.
    Obviously,~$\skipp(1)$ and~$\board(1)$ are not improving for~$\pi_0$, so assume~$x\neq 0$ in the following.
    In Lemma~\ref{Bland applies even can phases}, we argue that~$\travel(2)$,~$\skipp(1)$ and~$\leave(1)$ are active in~$\pi_x$ if~$\ell_1(x)=2$. 
    Then,~$\enter(1)$ is clearly more improving than~$\board(1)$.
    Finally, if~$\ell_1(x)\geq3$, we argue that~$\board(1)$ is active in~$\pi_x$.
    Further, property~<d> yields that either~$\board(2)$ or~$\leave(1)$ is active in~$\pi_x$.
    In the first case,~$\skipp(1)$ is not improving at all, and in the second case, it is not more beneficial than~$\enter(1)$.
    This concludes the proof.    
\end{proof}

\begin{restatable}{observation}{obsonetravel}\label{only one point improving}
    At any point during the execution of~$\textsc{Bland}(\mbl,\pi_0,\order)$, at most one of the edges~$\travel(i)$ with~$i\in[n]$ is improving.    
\end{restatable}
\begin{proof}
    If there are one or more improving travel edges at some point, the Bland numbering~$\order$ yields that the algorithm immediately applies one of them (the one traveling to the lowest level).
    This only occurs during the third canonical phase and after the switch~$\enter(1)$ got applied to an even canonical policy.
    
    For the first case, we already argued in the proof of Lemma~\ref{Bland applies odd can phases} that the applied travel edge is the only improving travel edge.
    For the second case, recall that, by Lemma~\ref{no pointer improving}, travel edges are not improving for canonical policies.
    As~$\travel(1)$ is inactive in even canonical policies, the switch~$\enter(1)$ does not affect the reduced costs of any travel edge besides~$\travel(1)$.
    Thus, this is the only travel edge becoming improving, which completes the proof.
\end{proof}

\section{A Combined Exponential Bound}
\label{sec:Dantzig}

In this section, we consider a family~$(\mdz=(V^A_n,V^R_n,E^A_n,E^R_n,r_{\mdz},p_{\mdz}))_{n\in\N}$ of Markov decision processes such that each process~$\mdz$ results from the process~$\mbl=(V_{\mbl},E_{\mbl},r_{\mbl})$ of the previous section by replacing every edge, besides the sink-loop, with the construction given in Figure~\ref{fig:Dantzig replaced};
note that the construction uses a probability~$p_v\in(0,1]$ for every~$v\in V_{\mbl}\setminus\{s\}$, which we will choose later.

\begin{figure}[h]
    \centering
    \resizebox{.8\textwidth}{!}{%
\begin{tikzpicture}[node distance = 3.2cm,
      p0/.style={circle, draw=blue!60, fill=blue!15, thick, minimum size=12mm},
      bl/.style={fill=white,draw=gray!30,text=gray,font=\footnotesize},
      every path/.style={->,draw=blue,line width=1pt,{>=Latex}},
      p1/.style={rectangle, draw=red!60, fill=red!15, thick, minimum size=12mm}]
      
    \node[p0] (v) {$v$};
    
    \node[p0,right of=v] (x) {$x_{v,w}$};
    \draw[bend left] (v) edge (x);
    \draw[bend left] (x) edge (v);

    \node[p1, right of=x] (y) {$y_{v,w}$};
    \draw[out=135,in=45] (y) edge[red] node[above,red] {$1-p_v$} (v);
    \draw (x) -- (y);

    \node[p0,right of=y] (z) {$z_{v,w}$};
    \draw[red] (y) edge[red] node[above,red] {$p_v$} (z);

    \node[p0,right of=z] (w) {$w$};
    \draw (z) edge node[above,blue] {$r_{\mathcal{B}_n}((v,w))$} (w);
\end{tikzpicture}%
}
      \caption{The construction that replaces every edge~$(v,w)\in E_{\mbl}\setminus\{(s,s)\}$ in~$\mdz$. Circular vertices are agent vertices, square ones are randomization vertices. Edge labels denote rewards and probabilities, where~$p_v\in(0,1]$. Note that, since~$v$ and~$w$ remain in the process, we have~$V_{\mbl}\subseteq V^A_n$.}
      \label{fig:Dantzig replaced}
\end{figure}
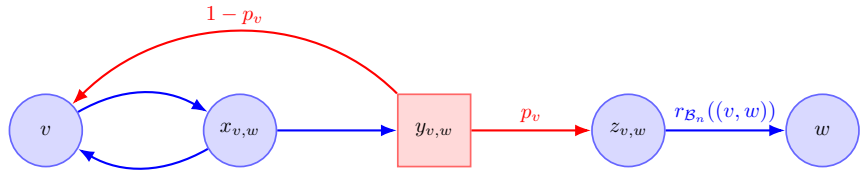

In the following, consider~$\mdz$ for some arbitrary but fixed~$n\in\N$.
The aim of this section is to show that policy iteration with Bland's rule, with Dantzig's rule, and with the Largest Increase rule performs~$\Omega(2^n)$ improving switches to a suitable initial policy for~$\mdz$.

Before we can analyze the behavior of \textsc{Bland} on~$\mdz$, we need to specify the Bland numbering~$\dorder\colon E^A_n\to|E^A_n|$ for~$\mdz$.
It is constructed as follows:
starting from the numbering~$\mathcal{N}_{\mbl}$, replace every edge~$(v,w)\in E_{\mbl}\setminus\{(s,s)\}$ by the edges~$(x_{v,w},y_{v,w})$ and~$(v,x_{v,w})$. 
Then, insert all edges of the form~$(x_{u,\cdot},u)$ with~$u\in V_{\mbl}\setminus\{s\}$ at the beginning of the numbering (the internal order of these edges can be chosen arbitrarily).
We do not need to specify the Bland numbers of edges that are the unique outgoing edge of a vertex.
    
Now that we have a Bland numbering, we want to transfer our results from the previous section to the new Markov decision process~$\mdz$. 
The following definition extends policies for~$\mbl$ to policies for~$\mdz$.

\begin{definition2}
    Let~$\pi$ and~$\pi'$ be policies for~$\mbl$ and~$\mdz$, respectively, and let~$v\in V_{\mbl}\setminus\{s\}$.
    Assume there is a~$w\in\Gamma_{\mbl}^+(v)$ such that~$(v,x_{v,w})$,~$(x_{v,w},y_{v,w})$, and~$(x_{v,u},v)$ are active in~$\pi'$ for all~$u\in\Gamma_{\mbl}^+(v)\setminus\{w\}$.
    Then, we say that~$v$ is \emph{($w$-)oriented w.r.t.\ $\pi'$}.
    We call~$\pi'$ the \emph{twin policy} of~$\pi$ if every vertex~$v\in V_{\mbl}\setminus\{s\}$ is~$\pi(v)$-oriented w.r.t.\ $\pi'$. \lipicsEnd
\end{definition2}

Let~$\pi_0'$ denote the twin policy of the canonical policy~$\pi_0$ for~$\mbl$. 
We could start by showing that~$\textsc{Bland}(\mdz,\pi_0',\dorder)$ visits the twin policy of every policy that~$\textsc{Bland}(\mbl,\pi_0,\order)$ visits.
Note that this would immediately imply the desired exponential number of improving switches.
However, we prefer to gather some general results first, which then allows for a more unified treatment of the three pivot rules.

Starting in a~$w$-oriented vertex~$v\in V_{\mbl}\setminus\{s\}$, the agent reaches vertex~$w$ with probability one (due to~$p_v>0$), while collecting a reward of~$r_{\mbl}((v,w))$. 
This immediately yields the following result.

\begin{observation}\label{same values}
    Let~$\pi$ be a policy for~$\mbl$ with twin policy~$\pi'$ for~$\mdz$.
    Then, for every vertex~$v\in V_{\mbl}$, we have~$\Valsolo_{\pi,\mbl}(v)=\Valsolo_{\pi',\mdz}(v)$.
\end{observation}

By the same argument, twin policies of weak unichain policies are weak unichain, and the proof idea of Lemma~\ref{opt for B} carries over.

\begin{observation}\label{opt for D}
    The twin policy of every weak unichain policy for~$\mbl$ is a weak unichain policy for~$\mdz$.
    The twin policy of the optimal policy for~$\mbl$ is optimal for~$\mdz$.
\end{observation}

By Theorem~\ref{thm: PI visits only wu}, this guarantees the correctness of~$\textsc{PolicyIteration}(\mdz,\pi_0')$.
Further, Theorem~\ref{thm:the connection} will allow us to carry our results over to the simplex method.

Since twin policies are central in our analysis, it comes in handy that only a certain type of edges might be improving for them.

\begin{observation}\label{only xy improving for twin}
	Let~$\pi'$ be the twin policy of some policy for~$\mbl$.
	Then, all improving switches for~$\pi'$ are of the form~$(x_{v,w},y_{v,w})\in E^A_n$ for some~$(v,w)\in E_{\mbl}$.
\end{observation}
\begin{proof}
	Since every vertex~$u\in V_{\mbl}\setminus\{s\}$ is oriented w.r.t\ $\pi'$, edges of the form~$(x_{u,\cdot},u)$ or~$(u,x_{u,\cdot})$ are either active or their application creates a zero-reward cycle of length two.
    Hence, none of these edges is improving \mbox{for~$\pi'$}.
\end{proof}

The following Lemma shows how the probabilities~$(p_v)_{v\in V_{\mbl}\setminus\{s\}}$ affect the reduced costs of these potentially improving edges.
Further, it yields a connection between the improving switches for a policy for~$\mbl$ and those for its twin policy.

\begin{restatable}{lemma}{lemimprBD}\label{reduced costs Bl Dz}
    Let~$\pi$ be a policy for~$\mbl$ with twin policy~$\pi'$, and let~$(v,w)\in E_{\mbl}\setminus\{(s,s)\}$.
    Then,~$z_{\pi',\mdz}(x_{v,w},y_{v,w})=p_v\cdot z_{\pi,\mbl}(v,w)$.
    In particular,~$(x_{v,w},y_{v,w})$ is improving for~$\pi'$ if and only if~$(v,w)$ is improving \mbox{for~$\pi$}.
\end{restatable}
\begin{proof}
    For convenience, we write~$x$,~$y$, and~$z$ instead of~$x_{v,w}$,~$y_{v,w}$, and~$z_{v,w}$.
    If~$(v,w)$ is active in~$\pi$, vertex~$v$ is~$w$-oriented w.r.t.\ $\pi'$. 
    Thus,~$(x,y)$ is active in~$\pi'$ as well.
    Hence, both edges are not improving as they have reduced costs of~$z_{\pi'}(x,y)=z_{\pi}(v,w)=0$. 
    
    Now assume that~$(v,w)$ is inactive \mbox{in~$\pi$}, which yields that~$(x,v)$ is active in~$\pi'$.
    We obtain
    \begin{gather}
    \begin{aligned}\label{rc of xy}
    z_{\pi',\mdz}(x,y)&=\Valsolo_{\pi'}(y)-\Valsolo_{\pi'}(x)=\Valsolo_{\pi'}(y)-\Valsolo_{\pi'}(v) \\
    &=p_v\Valsolo_{\pi'}(z)+(1-p_v)\Valsolo_{\pi'}(v)-\Valsolo_{\pi'}(v) \\
    &=p_v(\Valsolo_{\pi'}(w)+r_{\mbl}((v,w))-\Valsolo_{\pi'}(v))=p_v z_{\pi,\mbl}(v,w),  
    \end{aligned}
    \end{gather}
    where we used Observation~\ref{same values} for the last equality. The equivalence holds since~$p_v>0$.
\end{proof}

Note that we can transform a given twin policy with three switches into a different one by changing the orientation of an agent vertex~$v\in V_{\mbl}\setminus\{s\}$.
The following Lemma shows that, if applied consecutively, these switches all have the same reduced costs.

\begin{restatable}{lemma}{lemsamerc}\label{red costs erben}
    Let~$(v,w)\in E_{\mbl}\setminus\{(s,s)\}$ and let the policy~$\pi$ for~$\mdz$ be the twin policy of some weak unichain policy for~$\mbl$.
    If the edge~$(x_{v,w},y_{v,w})$ is improving for~$\pi$, we have
    \[
    z_\pi(x_{v,w},y_{v,w})=z_{\pi'}(v,x_{v,w})=z_{\pi''}(\pi(v),v),
    \]
where~$\pi'$ denotes the policy that results from applying~$(x_{v,w},y_{v,w})$ to~$\pi$ and~$\pi''$ denotes the policy that results from applying~$(v,x_{v,w})$ to~$\pi'$.
\end{restatable}
\begin{proof}
    We write~$x$,~$y$, and~$z$ instead of~$x_{v,w}$,~$y_{v,w}$, and~$z_{v,w}$.
    Assume that~$(x,y)$ is improving for~$\pi$.
    Then,~$(x,v)$ is active and~$z_\pi(x,y)=p_v(\Val(w)+r_{\mbl}((v,w))-\Val(v))$ by equation~\eqref{rc of xy}.
    
    Since~$(x,v)$ is active in~$\pi$, we know that~$(v,x)$ is inactive.    
    Hence, the value of~$v$ does not change during the switch to~$(x,y)$, i.e.,~$\Valsolo_{\pi'}(v)=\Val(v)$.
    As the value of~$w$ does not change either, we obtain
    \begin{align*}
        z_{\pi'}(v,x) &=\Valsolo_{\pi'}(x)-\Valsolo_{\pi'}(v)= \Valsolo_{\pi'}(y)-\Valsolo_{\pi'}(v) \\
        &=p_v(\Valsolo_{\pi'}(w)+r_{\mbl}((v,w))-\Valsolo_{\pi'}(v))=z_{\pi}(x,y).
    \end{align*}
    
    The vertex~$v$ is~$u$-oriented w.r.t.\ $\pi$ for some~$u\in\Gamma^+_{\mbl}(v)\setminus\{w\}$.
    Thus, we have~$\pi(v)=x_{v,u}$ and~$\Val(v)=\Val(u)+r_{\mbl}((v,u))$.
    Policy~$\pi$ is the twin policy of a weak unichain policy~$\bar\pi$, in which~$(v,u)$ is active.
    Therefore,~$\bar\pi$ does not contain a path from~$u$ to~$v$ in~$\mbl$, so~$\pi$ also does not reach~$v$ from~$u$ in~$\mdz$.
    This yields~$\Valsolo_{\pi''}(u)=\Val(u)$ and thus
    \begin{equation}\label{eq:help1}
    		\Val(v)=\Valsolo_{\pi''}(u)+r_{\mbl}((v,u)).
    \end{equation}
    The policy~$\pi$ is weak unichain, so Theorem~\ref{thm: PI visits only wu} yields that~$\pi''$ is weak unichain as well. 
    Since all vertices except~$v$ are oriented in~$\pi''$ and since~$\pi''$ walks from~$v$ to~$w$ with probability one, this yields that~$\pi''$ does not reach~$v$ from~$w$.
    The same thus holds for~$\pi$, so
    \begin{equation}\label{eq:help2}
    		\Valsolo_{\pi''}(w)=\Val(w).
    \end{equation}
    With this, we obtain
    \begin{align*}
        z_{\pi''}(\pi(v),v) &=\Valsolo_{\pi''}(v)-\Valsolo_{\pi''}(\pi(v))= \Valsolo_{\pi''}(v)-\Valsolo_{\pi''}(y_{v,u}) \\
        &=\Valsolo_{\pi''}(v)-p_v(\Valsolo_{\pi''}(u)+r_{\mbl}((v,u)))-(1-p_v)\Valsolo_{\pi''}(v)\\
        &\overset{\eqref{eq:help1}}{=}p_v(\Valsolo_{\pi''}(v)-\Val(v))\\
        &=p_v(\Valsolo_{\pi''}(w)+r_{\mbl}((v,w))-\Val(v))\\
        &\overset{\eqref{eq:help2}}{=}p_v(\Val(w)+r_{\mbl}((v,w))-\Val(v))\\
        &=z_{\pi}(x,y),
    \end{align*}
    which completes the proof.
\end{proof}

It is essential for the proofs of Lemma~\ref{Bland applies even can phases} and Lemma~\ref{Bland applies odd can phases} that~$\textsc{Bland}(\mbl,\order)$ prefers switches in vertices appearing early in the \emph{vertex numbering}~$\mathcal{N}_V\colon V_{\mbl}\to|V_{\mbl}|$ given by~$(t,a_1,b_1,a_2,b_2,\dots,a_n,b_n,d,s)$, i.e., let~$\mathcal{N}_V(t)=1$,~$\mathcal{N}_V(a_1)=2$, and so on.
Using the following definition, we can observe a similar behavior of policy iteration with Dantzig's rule, cf.\ Algorithm~\ref{alg:Dantzig}, on~$\mdz$.

\begin{definition2}\label{Def belong}
    The edge~$e\in E^A_n$ \emph{belongs} to vertex~$v\in V_{\mbl}\setminus\{s\}$ if 
    \begin{equation*} 	e\in\B(v)\coloneqq\bigcup_{w\in\Gamma^+_{\mbl}(v)}\{(x_{v,w},y_{v,w}),(v,x_{v,w}),(x_{v,w},v)\}.\lipicsEnd
	\end{equation*}     
\end{definition2}

\begin{algorithm}[b]
    \caption{$\textsc{Dantzig}(G,\pi)$}\label{alg:Dantzig}
    \vgap
    \textbf{input: }Markov decision process $G$, weak unichain policy $\pi$ for~$G$
    \vgap
    \hrule  
    \vgap
    \While{$\pi$ admits an improving switch}
    {
       ~$\bar s\gets \text{an improving switch~$s$ for~$\pi$ maximizing~$z_\pi(s)$}$\;
        \vgap
       ~$\pi\gets \nwstr{\bar s}$
    }
    \textbf{return}~$\pi$
\end{algorithm}

We obtain the following bounds on the reduced costs.

\begin{restatable}{lemma}{lemboundedrc}\label{red costs bounded in D}
    Let~$v\in V_{\mbl}\setminus\{s\}$ and~$e\in B(v)$ be arbitrary.
    Let~$\pi$ be a weak unichain policy for~$\mdz$ such that all vertex values w.r.t.\ $\pi$ are non-negative.
    If~$e$ is improving for~$\pi$, then its reduced costs are bounded by~$p_v\cdot0.25\leq z_\pi(e)\leq p_v\cdot 2^{n+2}$.
\end{restatable}
\begin{proof}
    Since~$e\in B(v)$, we have~$e\in\{(x_{v,w},y_{v,w}),(v,x_{v,w}),(x_{v,w},v)\}$ for some~$w\in\Gamma^+_{\mbl}(v)$.
    For convenience, we write~$x$,~$y$, and~$z$ instead of~$x_{v,w}$,~$y_{v,w}$, and~$z_{v,w}$.
    
    Firstly, assume that~$e=(x,y)$. 
    Then, since~$e$ is improving,~$(x,v)$ is active in~$\pi$.
    As in equation~\eqref{rc of xy}, we obtain~$z_{\pi}(x,y)=p_v(\Val(w)+r_{\mbl}((v,w))-\Val(v))\eqqcolon p_v\cdot\delta(\pi,v,w)$.

    Secondly, assume that~$e=(v,x)$.
    Then,~$(x,y)$ is active in~$\pi$ as otherwise~$e$ would not be improving due to~$z_\pi(v,x)=\Val(x)-\Val(v)=0$.
    This yields
    \begin{equation}\label{rc of vx}
    z_{\pi}(v,x)=\Val(x)-\Val(v)=\Val(y)-\Val(v)=p_v\cdot\delta(\pi,v,w).
    \end{equation}

    Lastly, assume~$e=(x,v)$. 
    Then, as before,~$(x,y)$ is active in~$\pi$.
    We can thus conclude from~\eqref{rc of vx} that~$z_{\pi}(x,v)=\Val(v)-\Val(x)=-p_v\cdot\delta(\pi,v,w)$.

    By assumption, every vertex has a non-negative value with respect to~$\pi$.
    Further, all vertex values are bounded from above by the maximum vertex value w.r.t.\ the optimal policy for~$\mdz$. 
    By Lemma~\ref{opt for B}, Observation~\ref{same values}, and Observation~\ref{opt for D}, this is~$\Valsolo_{\pi_*'}(t)=2^{n+1}-1.25$.
    
    Since the absolute value of any edge reward is at most~$2^n$, we obtain 
    \[
    |\delta(\pi,v,w)| \leq \left(2^{n+1}-1.25+2^n\right)\leq 2^{n+2}.
    \]
    Hence, we have an upper bound of~$z_\pi(e)\leq p_v\cdot 2^{n+2}$.
    
    As all edge rewards are integer multiples of~0.25, also~$\Val(u)$ is an integer multiple of~0.25 for every~$u\in V_{\mbl}$ (starting in~$u$, policy~$\pi$ visits every edge that has a non-zero reward either exactly once or never).
    This yields that~$|\delta(\pi,v,w)|$ is a multiple of 0.25 as well, which concludes the proof.
\end{proof}

Note that, by Theorem~\ref{thm: PI visits only wu} and as all vertex values w.r.t.\ $\pi_0'$ are zero, $\textsc{Dantzig}(\mdz,\pi_0')$ only visits weak unichain policies with non-negative vertex values.
Therefore, if we choose the probabilities~$(p_v)_{v\in V_{\mbl}\setminus\{s\}}$ for increasing vertex numbers~$\mathcal{N}_V(v)$ fast enough decreasing, then the previous lemma yields that~\textsc{Dantzig} prefers improving switches that belong to vertices appearing early in the vertex numbering~$\mathcal{N}_V$.
We use this in the proof of Theorem~\ref{Dantzig exp on D}.

The following technical result holds independently of the chosen pivot rule.

\begin{restatable}{lemma}{xvbeforesuc}\label{xv before succeeding}
	Let~$(v,w)\in E_{\mbl}$ be an improving switch that gets applied to a policy~$\pi$ during the execution of~$\textsc{Bland}(\mbl,\pi_0,\order)$.
	Let~$\bar\pi'$ denote the policy for~$\mdz$ that results from applying the switches~$(x_{v,w},y_{v,w})$ and~$(v,x_{v,w})$ to the twin policy~$\pi'$ of~$\pi$.
	Then, the edge~$(x_{v,\pi(v)},v)$ is improving for~$\bar\pi'$ and it remains improving during the execution of~$\textsc{PolicyIteration}(\mdz,\bar\pi')$ until it gets applied or until an improving switch of the form~$(x_{u,\cdot},y_{u,\cdot})$ with~$\mathcal{N}_V(u)>\mathcal{N}_V(v)$ gets applied.
\end{restatable}
\begin{proof}
	The policy~$\pi$ is weak unichain due to Theorem~\ref{thm: PI visits only wu}.
	Further, according to Lemma~\ref{reduced costs Bl Dz}, the edge~$(x_{v,w},y_{v,w})$ is improving for~$\pi'$.
	Thus, Lemma~\ref{red costs erben} yields that~$(x_{v,\pi(v)},v)$ is improving for~$\bar\pi'$, which concludes the first part of the proof.
	
	Note that every vertex~$u$ succeeding~$v$ in~$\mathcal{N}_V$ is~$\pi(u)$-oriented	w.r.t.\ $\bar\pi'$.
	Therefore, as in Observation~\ref{only xy improving for twin}, the first improving switch getting applied by~$\textsc{PolicyIteration}(\mdz,\bar\pi')$ and belonging to a vertex~$u$ succeeding~$v$ in~$\mathcal{N}_V$ must be of the form~$(x_{u,\cdot},y_{u,\cdot})$.
	
	 We claim that the edge~$(x_{v,\pi(v)},v)$ remains improving until it gets applied or until~$\textsc{PolicyIteration}(\mdz,\bar\pi')$ applies an improving switch belonging to some vertex~$u$ with~$\mathcal{N}_V(u)>\mathcal{N}_V(v)$, which concludes the proof.
	 To verify this claim, we perform a case analysis on the specific improving switch $(v,w)\in E_{\mbl}$ that gets applied to $\pi$.
	
	Assume that~$(v,w)=\enter(1)$ is the improving switch that \textsc{Bland} applies, by Lemma~\ref{Bland applies even can phases}, to the canonical policy~$\pi$ for any even~$x\in[2^n-2]_0$.
    If~$x=0$, then~$(x_{a_1,\pi(a_1)},a_1)$ is the only improving switch for~$\bar{\pi}'$ that belongs to~$a_1$ or~$t$.
    Therefore, either~$(x_{a_1,\pi(a_1)},a_1)$ or a switch belonging to a vertex succeeding~$a_1$ in~$\mathcal{N}_V$ gets applied to~$\bar{\pi}'$.
    If~$x\neq0$, then, by Lemma~\ref{Bland applies even can phases} and Lemma~\ref{only one point improving}, the edge~$\travel(1)$ is the only improving switch in $t$ after the application of~$\enter(1)$.
    Therefore, it remains more beneficial to enter level 1 than to skip it or to board~$t$ from~$a_1$ as long as no switch in a vertex succeeding~$a_1$ in~$\mathcal{N}_V$ gets applied.
    Hence, the claim holds in this first case that~$(v,w)=\enter(1)$.
    
    Assume now that~$(v,w)=\travel(i)$ for some~$i\in[n]$.
    Then, Observation~\ref{only one point improving} and Lemma~\ref{reduced costs Bl Dz} yield that~$(x_{t,\pi(t)},t)$ is the only improving switch belonging to~$t$.
    Since~$t$ is the first vertex in~$\mathcal{N}_V$, this verifies the claim.
    
    Recall that, by Lemma~\ref{Bland applies even can phases},~$\enter(1)$ and~$\travel(1)$ are the only improving switches that might get applied during the transition from~$\pi_x$ to~$\pi_{x+1}$ for even~$x$.
    Hence, it now suffices to prove the claim for the case that~$(v,w)$ is part of the canonical phases w.r.t.\ some odd~$x\in[2^n-3]$ with least unset bit~$\ell \coloneqq \ell_0(x)>1$.
    The following analysis relies on the structural insights from the proof of Lemma~\ref{Bland applies odd can phases}.

    If~$(v,w)=\leave(i)$ is the switch from the first or the last phase, then~$\enter(i+1)$ and~$\stay(i)$ are active in~$\pi$.
    Provided that level~$i+1$ gets entered, leaving level~$i$ is more beneficial than staying in it, which verifies the claim in this case.

    If~$(v,w)=\stay(\ell-1)$, then~$(x_{v,\pi(v)},v)$ is the only improving switch for$\bar\pi'$ in the first~$\ell-1$ levels, which again validates the claim.

    If~$(v,w)=\enter(\ell)$, then Observation~\ref{red costs in ai} implies that~$\skipp(\ell)$ and~$\board(\ell)$ are not improving for the policy~$\pi^{(v,w)}$, i.e.,~$(x_{a_\ell,\pi(a_\ell)},a_\ell)$ is the only improving switch belonging to~$a_\ell$.
    Since entering level~$\ell$ is more beneficial than entering all of the first~$\ell-1$ levels, this holds true as long as the algorithm only applies switches in levels below~$\ell$.
    Thus, the claim holds in this case.

    If~$(v,w)=\board(j)$ for some~$j\in[\ell-2]$, then~$(x_{v,\pi(v)},v)$ is the only improving switch for$\bar\pi'$ in the first~$j$ levels, verifying the claim.

    If~$(v,w)=\skipp(\ell-1)$, then~$\stay(\ell-1)$ and~$\enter(\ell)$ are active in~$\pi$.
    As long as these edges stay active, skipping level~$\ell-1$ is more beneficial than entering it.
    Further,~$\travel(\ell)$ is active in~$\pi$, so skipping level~$\ell-1$ remains better than boarding~$t$ from~$a_{\ell-1}$ until a switch in some higher level has made a travel edge~$\travel(k)$ with $k>\ell$ improving. This supports the claim once again.

    Finally, assume~$(v,w)=\stay(j)$ for some~$j\in[\ell-3]$. This switch is improving for~$\pi$ as it allows the policy to reach and enter level~$\ell$ without using~$\board(j+1)$ and~$\travel(\ell)$.
    An edge~$\travel(k)$ with~$k>\ell$ can only become improving after an improving switch in some higher level.
    Therefore, staying in level~$j$ is better than leaving it as long as no switches in levels above~$j$ get applied. 
    Hence, our claim holds in all cases, which concludes the proof.	
\end{proof}

With this, we can show that a certain class of pivot rules, including Bland's, Dantzig's, and the Largest Increase rule, yield an exponential number of improving switches on~$\mdz$.

\begin{restatable}{lemma}{prconditions}\label{pr conditions}
	Assume that~$\textsc{PolicyIteration}(\mdz,\pi_0')$, where~$\pi_0'$ denotes the twin policy of~$\pi_0$, gets applied with a pivot rule that satisfies the following conditions:
	\begin{enumerate}[(a)]
	\item For every improving switch~$(v,w)\in E_{\mbl}$ that~$\textsc{Bland}(\mbl,\pi_0,\order)$ applies to some policy~$\pi$,~$\textsc{PolicyIteration}$ applies~$(x_{v,w},y_{v,w})$ and~$(v,x_{v,w})$ to the twin policy of~$\pi$.
	\item While an edge of the form~$(x_{v,\cdot},v)$ is improving for some~$v\in V_{\mbl}$,~$\textsc{PolicyIteration}$ does not apply an improving switch of the form~$(x_{u,\cdot},y_{u,\cdot})$ with~$\mathcal{N}_V(u)>\mathcal{N}_V(v)$.
	\end{enumerate}
	Then,~$\textsc{PolicyIteration}(\mdz,\pi_0')$ performs~$\Omega(2^n)$ improving switches.
\end{restatable}
\begin{proof}
Let~$\pi$ be a policy for~$\mbl$ occurring during the execution of~$\textsc{Bland}(\mbl,\pi_0,\order)$, where we allow~$\pi=\pi_0$, and let~$(v,w)\in E_{\mbl}$ denote the switch that \textsc{Bland} applies to~$\pi$.
    Let~$\pi'$ be the twin policy \mbox{of~$\pi$}, and let~$\tilde{\pi}=\pi^{(v,w)}$.  
    
    By condition (a),~$\textsc{PolicyIteration}$ applies~$(x_{v,w},y_{v,w})$ and~$(v,x_{v,w})$ to~$\pi'$.
    Denote the resulting policy by~$\bar\pi'$.
    
    According to Lemma~\ref{xv before succeeding}, the edge~$(x_{v,\pi(v)},v)$ now stays improving until it gets applied as an improving switch or until an improving switch of the form~$(x_{u,\cdot},y_{u,\cdot})$ with~$\mathcal{N}_V(u)>\mathcal{N}_V(v)$ gets applied.
    With condition (b), this yields that~$(x_{v,\pi(v)},v)$ gets applied by \textsc{PolicyIteration} at some point, and that it is constantly improving until then.

    Note that, as long as~$(v,x_{v,\pi(v)})$ is inactive, the policy's choice in~$x_{v,\pi(v)}$ only affects the reduced costs of its unique incidental edge~$(v,x_{v,\pi(v)})$.
    This edge is not active in~$\bar\pi'$ and is not improving until the application of~$(x_{v,\pi(v)},v)$.
    Therefore, if we were to force the algorithm to apply~$(x_{v,\pi(v)},v)$ to~$\bar\pi'$, this would not alter the remaining behavior of the algorithm.
    The policy resulting from this forced switch is the twin policy of~$\Tilde{\pi}$.

    Hence, without changing the total number of applied improving switches (we only rearrange them), we can assume that~$\textsc{PolicyIteration}(\mdz,\pi_0')$ visits the twin policy of every policy visited by~$\textsc{Bland}(\mbl,\pi_0,\order)$.
    By Theorem \ref{exp auf B}, this yields that the algorithm needs to perform an exponential number of improving switches, which concludes the proof.
\end{proof}

Now it suffices to check the conditions given in Lemma~\ref{pr conditions} for each pivot rule.

\begin{proposition}\label{Bland exp on D}
    Let~$\pi_0'$ denote the twin policy of~$\pi_0$.
    Then,~$\textsc{Bland}(\mdz,\pi_0',\dorder)$ performs~$\Omega(2^n)$ improving switches.
\end{proposition}
\begin{proof}
	We check the two conditions from Lemma~\ref{pr conditions}.
    For condition (a), let~$\pi$ be a policy for~$\mbl$ visited by~$\textsc{Bland}(\mbl,\pi_0,\order)$, including the case~$\pi=\pi_0$, and let~$\pi'$ be the twin policy of~$\pi$.
    Assume that \textsc{Bland} applies the improving switch~$(v,w)\in E_{\mbl}$ to~$\pi$.
    
	By Observation~\ref{only xy improving for twin}, all improving switches for~$\pi'$ are of the form~$(x_\cdot,y_\cdot)$.
    According to Lemma~\ref{reduced costs Bl Dz}, the edge~$(x_{v,w},y_{v,w})$ is improving for~$\pi'$. 
    As~$(v,w)$ is the improving switch for~$\pi$ with the smallest Bland number in~$\order$, we know that, by construction of~$\dorder$, the algorithm applies the switch~$(x_{v,w},y_{v,w})$ to~$\pi'$.
   
    Further, since~$\pi$ is weak unichain due to Theorem~\ref{thm: PI visits only wu}, Lemma~\ref{red costs erben} yields that~$(v,x_{v,w})$ is improving after this switch.
    As it is the successor of~$(x_{v,w},y_{v,w})$ in~$\dorder$ and as no other egde became improving due to the first switch, the algorithm applies~$(v,x_{v,w})$ next.
    That is, Bland's rule satisfies condition (a).
    
    Additionally, condition (b) holds since the edge~$(x_{v,\pi(v)},v)$ precedes any switch of the form~$(x_{u,\cdot},y_{u,\cdot})$ with~$\mathcal{N}_V(u)>\mathcal{N}_V(v)$ in the Bland numbering~$\dorder$.    
\end{proof}

As motivated above, the choice of the probabilities~$(p_v)_{v\in V_{\mbl}\setminus\{s\}}$ in the following theorem yields that~\textsc{Dantzig} prefers improving switches that belong to vertices appearing early in the vertex numbering~$\mathcal{N}_V$.

\begin{restatable}{theorem}{thmdantzigexp}\label{Dantzig exp on D}
    Let~$p_v=2^{-\mathcal{N}_V(v)(n+5)}$ for all~$v\in V_{\mbl}\setminus\{s\}$, and let~$\pi_0'$ denote the twin policy of~$\pi_0$. Then,~$\textsc{Dantzig}(\mdz,\pi_0')$ performs~$\Omega(2^n)$ improving switches. 
\end{restatable}
\begin{proof}
	We check the two conditions from Lemma~\ref{pr conditions}.
    For condition (b), we compute that the choice of the probabilities~$p_v$ yields that~\textsc{Dantzig} prefers improving switches belonging to vertices with a small vertex number.
	
	Let~$u,v \in V_{\mbl}\setminus\{s\}$ with~$\mathcal{N}_V(u)>\mathcal{N}_V(v)$. 
	Let further~$e_u\in B(u)$ and~$e_v\in B(v)$ be improving switches for some policy~$\pi$ for~$\mdz$, which gets visited by~$\textsc{PolicyIteration}(\mdz,\pi_0')$.
	Then,~$\pi$ is weak unichain and only induces non-negative vertex values, so Lemma~\ref{red costs bounded in D} yields
	\[
	z_{\pi}(e_v)\geq p_v\cdot 0.25=2^{-\mathcal{N}_V(v)(n+5)-2}\geq 2^{-\mathcal{N}_V(u)(n+5)+(n+5)-2}=p_u\cdot 2^{(n+3)}> z_\pi(e_u).
	\]
	Hence, Dantzig's rule prefers switches belonging to~$v$ over those belonging to~$u$, so it satisfies condition (b).
	
	For condition (a), let~$\pi$ be a policy for~$\mbl$ visited by~$\textsc{Bland}(\mbl,\pi_0,\order)$, including the case~$\pi=\pi_0$, and let~$(v,w)\in E_{\mbl}$ denote the switch that \textsc{Bland} applies to~$\pi$.
    Let~$\pi'$ be the twin policy \mbox{of~$\pi$}.    

	By Observation~\ref{only xy improving for twin}, all improving switches for~$\pi'$ are of the form~$(x_{\cdot},y_{\cdot})$.
    By construction of~$\dorder$, we know that~$\textsc{Bland}$ prefers those switches~$(x_{\cdot},y_{\cdot})$ that belong to vertices with a small vertex number.
    In the proof of Proposition~\ref{Bland exp on D}, we see that~$\textsc{Bland}(\mdz,\pi',\dorder)$ applies~$(x_{v,w},y_{v,w})$ to~$\pi'$.
    Since \textsc{Dantzig} also prefers switches belonging to vertices with a small vertex number, we conclude that~$\textsc{Dantzig}(\mdz,\pi_0',\dorder)$ applies an improving switch to~$\pi'$ that belongs to~$v$.
    However, there might be multiple such switches.

    Recall that all improving switches for~$\pi'$ are of the form~$(x_{\cdot},y_{\cdot})$.
    If~$v=b_i$ for some~$i\in[n]$, then only two of these (possibly improving) edges belong to~$v$, one of which is active in~$\pi'$.
    Therefore, in this case, \textsc{Dantzig} applies the improving switch~$(x_{v,w},y_{v,w})$.

    Now assume~$v=a_i$ for some~$i\in[n]$.
    Then, Observation~\ref{red costs in ai} and Lemma~\ref{reduced costs Bl Dz} yield 
    \[
    z_{\pi'}(x_{a_i,b_i},y_{a_i,b_i}) > \max\{z_{\pi'}(x_{a_i,a_{i+1}},y_{a_i,a_{i+1}}), z_{\pi'}(x_{a_i,t},y_{a_i,t})\}
    \]
    if~$(v,w)=\enter(i)$, and
    \[
    z_{\pi'}(x_{a_i,a_{i+1}},y_{a_i,a_{i+1}}) > z_{\pi'}(x_{a_i,t},y_{a_i,t})
    \]
    if~$(v,w)=\skipp(i)$.
    Therefore, the edge~$(x_{v,w},y_{v,w})$ has higher reduced costs than the other edges that belong to~$v$, so \textsc{Dantzig} applies it to~$\pi'$.

    Finally, if~$v=t$, then Observation~\ref{only one point improving} and Lemma~\ref{reduced costs Bl Dz} yield that~$(x_{v,w},y_{v,w})$ is the only improving switch that belongs to~$t$.
    Thus, it gets applied by \textsc{Dantzig}.

    We conclude that, in all cases, \textsc{Dantzig} applies the switch~$(x_{v,w},y_{v,w})$ to~$\pi'$, which is the unique edge with the highest reduced costs.
   According to Lemma~\ref{red costs erben}, the edge~$(v,x_{v,w})$ has now the same reduced costs as~$(x_{v,w},y_{v,w})$ had before its application.
   Since~$(v,x_{v,w})$ is the only edge that became improving during the last switch, \textsc{Dantzig} applies this edge next.
   Therefore, Dantzig's rule also satisfies condition (a), which concludes the proof.
\end{proof}

\begin{algorithm}
    \caption{$\textsc{LargestIncrease}(G,\pi)$}\label{alg:LI}
    \vgap
    \textbf{input: }Markov decision process $G$, weak unichain policy $\pi$ for $G$
    \vgap
    \hrule  
    \vgap
    \While{$\pi$ admits an improving switch}
    {
       ~$\bar s\gets\text{an improving switch $s$ for $\pi$ maximizing }\sum_{v\in V^A_n}\Valsolo_{\pi^s}(v)$\;
        \vgap
       ~$\pi\gets \nwstr{\bar s}$
    }
    \textbf{return}~$\pi$
\end{algorithm}

Finally, we turn to policy iteration with the Largest Increase pivot rule, cf. Algorithm~\ref{alg:LI}.
In the most general sense, consider an arbitrary improving switch~$s=(v,w)$ for some policy~$\pi$.
Assume that no ingoing edges of~$v$ are active in~$\pi$.
Then, the reduced costs of~$s$ coincide with the increase of the sum over all vertex values, that is, we obtain the equality~$z_\pi(s)=\sum_{v\in V^A_n}\Valsolo_{\pi^s}(v)-\sum_{v\in V^A_n}\Valsolo_{\pi}(v)$.
Further, the induced increase of the sum is always at least as large as the reduced costs.

From this, using the structure of~$\mdz$, we can conclude that \textsc{LargestIncrease} mirrors the behavior of \textsc{Dantzig} if we choose the probabilities~$p_v$ as before.

\begin{restatable}{theorem}{LIexp}\label{LI exp on D}
     Let~$p_v=2^{-\mathcal{N}_V(v)(n+5)}$ for all~$v\in V_{\mbl}\setminus\{s\}$, and let~$\pi_0'$ denote the twin policy of~$\pi_0$. Then,~$\textsc{LargestIncrease}(\mdz,\pi_0')$ performs~$\Omega(2^n)$ improving switches.
\end{restatable}
\begin{proof}
    We check the two conditions from Lemma~\ref{pr conditions}. 
    For condition (a), let~$\pi$ be a policy for~$\mbl$ occurring during the execution of~$\textsc{Bland}(\mbl,\pi_0,\order)$, where we allow~$\pi=\pi_0$, and let~$(v,w)\in E_{\mbl}$ denote the switch that \textsc{Bland} applies to~$\pi$.
    Let~$\pi'$ be the twin policy \mbox{of~$\pi$}.
            
    According to the proof of Theorem~\ref{Dantzig exp on D}, \textsc{Dantzig} applies the improving switches~$(x_{v,w},y_{v,w})$ and~$(v,x_{v,w})$ to~$\pi'$.    
    By Observation~\ref{only xy improving for twin}, all improving switches for~$\pi'$ are of the form~$(x_{\cdot},y_{\cdot})$, where~$\pi'$ does not reach~$x_{\cdot}$ when starting in any other vertex.
    Therefore, since the probabilities~$(p_v)_{v\in V_{\mbl}\setminus\{s\}}$ are chosen as in Theorem~\ref{Dantzig exp on D}, the reduced costs of each of these improving switches coincide with the induced increase of the sum over all vertex values.
    Hence,~\textsc{LargestIncrease} also applies the improving switch~$(x_{v,w},y_{v,w})$ to~$\pi'$.

    This switch only increases the reduced costs of the edge~$(v,x_{v,w})$, which, by Lemma~\ref{red costs erben}, coincide with the previous reduced costs.
    Therefore, the induced increase of the sum over all vertex values is for~$(v,x_{v,w})$ now at least as large as it was for~$(x_{v,w},y_{v,w})$ before.
    Hence, \textsc{LargestIncrease} also applies~$(v,x_{v,w})$ next.
    We conclude that the Largest Increase pivot rule satisfies condition (a).
    
    Note that the reduced costs of the edges from condition (b) again coincide with the induced increase of the sum over all vertex values.
    Further, by the proof of Theorem~\ref{Dantzig exp on D}, we know that Dantzig's rule prefers switches belonging to vertices with a small vertex number.
    Therefore, the Largest Increase rule also satisfies condition (b).
\end{proof}

Note that Theorem~\ref{thm:main} is now a direct consequence of Theorems~\ref{Bland exp on D},~\ref{Dantzig exp on D}, and~\ref{LI exp on D}.
Moreover, we have seen in the proofs of these theorems that Bland's rule, Dantzig's rule, and the Largest Increase rule satisfy the conditions from Lemma~\ref{pr conditions}.
Thus, any combination of these rules also satisfies the conditions, which immediately yields Corollary~\ref{cor:comb}. Finally, Corollary~\ref{cor:main} follows from Theorem~\ref{thm:the connection} together with Obervation~\ref{pi0 is wu for B}, Lemma~\ref{opt for B}, and Observation~\ref{opt for D}.

\newpage

\bibliography{references}

\begin{thebibliography}{10}

\bibitem{adler1987simplex}
Ilan Adler, Richard~M Karp, and Ron Shamir.
\newblock A simplex variant solving an $m\times d$ linear program in
  $\mathcal{O} (\min (m^2, d^2))$ expected number of pivot steps.
\newblock {\em Journal of Complexity}, 3(4):372--387, 1987.

\bibitem{adler2014simplex}
Ilan Adler, Christos Papadimitriou, and Aviad Rubinstein.
\newblock On simplex pivoting rules and complexity theory.
\newblock In {\em Integer Programming and Combinatorial Optimization: 17th
  International Conference, IPCO 2014.}, pages 13--24. Springer, 2014.

\bibitem{amenta1999deformed}
Nina Amenta and G\"unter~M Ziegler.
\newblock Deformed products and maximal shadows of polytopes.
\newblock {\em Contemporary Mathematics}, 223:57--90, 1999.

\bibitem{avis1978notes}
David Avis and Vasek Chv{\'a}tal.
\newblock Notes on {B}land’s pivoting rule.
\newblock {\em Polyhedral Combinatorics: Dedicated to the memory of
  D.R.~Fulkerson}, pages 24--34, 1978.

\bibitem{avis2017exponential}
David Avis and Oliver Friedmann.
\newblock An exponential lower bound for {C}unningham’s rule.
\newblock {\em Mathematical Programming}, 161:271--305, 2017.

\bibitem{bach2025unconditionallowerboundactiveset}
Eleon Bach, Yann Disser, Sophie Huiberts, and Nils Mosis.
\newblock An unconditional lower bound for the active-set method in convex
  quadratic maximization, 2025.
\newblock \href {https://arxiv.org/abs/2507.16648} {\path{arXiv:2507.16648}}.

\bibitem{bellman1957dynamic}
Richard Bellman.
\newblock {\em Dynamic Programming}.
\newblock Princeton University Press, 1957.

\bibitem{bertsimas2004solving}
Dimitris Bertsimas and Santosh Vempala.
\newblock Solving convex programs by random walks.
\newblock {\em Journal of the ACM}, 51(4):540--556, 2004.

\bibitem{bland1977new}
Robert~G Bland.
\newblock New finite pivoting rules for the simplex method.
\newblock {\em Mathematics of Operations Research}, 2(2):103--107, 1977.

\bibitem{borgwardt1982average}
Karl-Heinz Borgwardt.
\newblock The average number of pivot steps required by the simplex-method is
  polynomial.
\newblock {\em Zeitschrift f{\"u}r Operations Research}, 26:157--177, 1982.

\bibitem{dadush2018friendly}
Daniel Dadush and Sophie Huiberts.
\newblock A friendly smoothed analysis of the simplex method.
\newblock In {\em Proceedings of the 50th Annual ACM Symposium on Theory of
  Computing (STOC)}, pages 390--403, 2018.

\bibitem{dantzig1963linear}
George Dantzig.
\newblock {\em Linear programming and extensions}.
\newblock Princeton university press, 1963.

\bibitem{dantzig1951maximization}
George~B. Dantzig.
\newblock Maximization of a linear function of variables subject to linear
  inequalities.
\newblock {\em Activity analysis of production and allocation}, 13:339--347,
  1951.

\bibitem{dantzig1955generalized}
George~B. Dantzig, Alex Orden, and Philip Wolfe.
\newblock The generalized simplex method for minimizing a linear form under
  linear inequality restraints.
\newblock {\em Pacific Journal of Mathematics}, 5(2):183--195, 1955.

\bibitem{deshpande2005improved}
Amit Deshpande and Daniel~A Spielman.
\newblock Improved smoothed analysis of the shadow vertex simplex method.
\newblock In {\em 46th Annual IEEE Symposium on Foundations of Computer Science
  (FOCS'05)}, pages 349--356. IEEE, 2005.

\bibitem{disser2023exponential}
Yann Disser, Oliver Friedmann, and Alexander~V. Hopp.
\newblock An exponential lower bound for zadeh’s pivot rule.
\newblock {\em Mathematical Programming}, 199(1-2):865--936, 2023.

\bibitem{disser2023unified}
Yann Disser and Nils Mosis.
\newblock A unified worst case for classical simplex and policy iteration pivot
  rules.
\newblock In {\em 34th International Symposium on Algorithms and Computation
  (ISAAC)}, pages 27--1, 2023.

\bibitem{disser2025unconditional}
Yann Disser and Nils Mosis.
\newblock An unconditional lower bound for the active-set method on the
  hypercube.
\newblock In {\em Proceedings of the 23rd Conference on Integer Programming and
  Combinatorial Optimization (IPCO)}, pages 213--227, 2025.

\bibitem{disser2018simplex}
Yann Disser and Martin Skutella.
\newblock The simplex algorithm is {NP}-mighty.
\newblock {\em ACM Transactions on Algorithms (TALG)}, 15(1):1--19, 2018.

\bibitem{dunagan2004simple}
John Dunagan and Santosh Vempala.
\newblock A simple polynomial-time rescaling algorithm for solving linear
  programs.
\newblock In {\em Proceedings of the 36th Annual ACM Symposium on Theory of
  Computing (STOC)}, pages 315--320, 2004.

\bibitem{Fearnley10}
John Fearnley.
\newblock Exponential lower bounds for policy iteration.
\newblock In {\em Proceedings of the 37th International Colloquium on Automata,
  Languages and Programming (ICALP)}, pages 551--562, 2010.

\bibitem{fearnley2015complexity}
John Fearnley and Rahul Savani.
\newblock The complexity of the simplex method.
\newblock In {\em Proceedings of the 47th Annual ACM Symposium on Theory of
  Computing (STOC)}, pages 201--208, 2015.

\bibitem{friedmann2011exponential}
Oliver Friedmann.
\newblock {\em Exponential lower bounds for solving infinitary payoff games and
  linear programs}.
\newblock PhD thesis, LMU Munich, 2011.

\bibitem{friedmann2011subexponential}
Oliver Friedmann, Thomas~D. Hansen, and Uri Zwick.
\newblock Subexponential lower bounds for randomized pivoting rules for the
  simplex algorithm.
\newblock In {\em Proceedings of the 43rd Annual ACM Symposium on Theory of
  Computing (STOC)}, pages 283--292, 2011.

\bibitem{gartner2002random}
Bernd G{\"a}rtner.
\newblock The random-facet simplex algorithm on combinatorial cubes.
\newblock {\em Random Structures \& Algorithms}, 20(3):353--381, 2002.

\bibitem{gartner2006linear}
Bernd G{\"a}rtner and Ingo Schurr.
\newblock Linear programming and unique sink orientations.
\newblock In {\em Proceedings of the 17th Annual ACM-SIAM Symposium on Discrete
  Algorithms (SODA)}, pages 749--757, 2006.

\bibitem{gass1955computational}
Saul Gass and Thomas Saaty.
\newblock The computational algorithm for the parametric objective function.
\newblock {\em Naval research logistics quarterly}, 2(1-2):39--45, 1955.

\bibitem{goldfarb1979worst}
Donald Goldfarb and William~Y. Sit.
\newblock Worst case behavior of the steepest edge simplex method.
\newblock {\em Discrete Applied Mathematics}, 1(4):277--285, 1979.

\bibitem{hansen2012worst}
Thomas~D. Hansen.
\newblock {\em Worst-case analysis of strategy iteration and the simplex
  method}.
\newblock PhD thesis, Aarhus University, 2012.

\bibitem{hansen2015improved}
Thomas~D. Hansen and Uri Zwick.
\newblock An improved version of the random-facet pivoting rule for the simplex
  algorithm.
\newblock In {\em Proceedings of the 47th Annual ACM Symposium on Theory of
  Computing (STOC)}, pages 209--218, 2015.

\bibitem{howard1960dynamic}
Ronald~A. Howard.
\newblock {\em Dynamic programming and {M}arkov processes}.
\newblock John Wiley, 1960.

\bibitem{huiberts2023upper}
Sophie Huiberts, Yin~Tat Lee, and Xinzhi Zhang.
\newblock Upper and lower bounds on the smoothed complexity of the simplex
  method.
\newblock In {\em Proceedings of the 55th Annual ACM Symposium on Theory of
  Computing (STOC)}, pages 1904--1917, 2023.

\bibitem{jeroslow1973simplex}
Robert~G. Jeroslow.
\newblock The simplex algorithm with the pivot rule of maximizing criterion
  improvement.
\newblock {\em Discrete Mathematics}, 4(4):367--377, 1973.

\bibitem{kalai1992subexponential}
Gil Kalai.
\newblock A subexponential randomized simplex algorithm.
\newblock In {\em Proceedings of the 24th Annual ACM Symposium on Theory of
  Computing (STOC)}, pages 475--482, 1992.

\bibitem{karmarkar1984new}
Narendra Karmarkar.
\newblock A new polynomial-time algorithm for linear programming.
\newblock In {\em Proceedings of the 16th Annual ACM Symposium on Theory of
  Computing (STOC)}, pages 302--311, 1984.

\bibitem{kelner2006randomized}
Jonathan~A Kelner and Daniel~A Spielman.
\newblock A randomized polynomial-time simplex algorithm for linear
  programming.
\newblock In {\em Proceedings of the 38th Annual ACM Symposium on Theory of
  Computing (STOC)}, pages 51--60, 2006.

\bibitem{khachiyan1980polynomial}
Leonid~G. Khachiyan.
\newblock Polynomial algorithms in linear programming.
\newblock {\em USSR Computational Mathematics and Mathematical Physics},
  20(1):53--72, 1980.

\bibitem{klee1972good}
Victor Klee and George~J Minty.
\newblock How good is the simplex algorithm?
\newblock {\em Inequalities}, 3(3):159--175, 1972.

\bibitem{matouvsek1992subexponential}
Ji{\v{r}}{\'\i} Matou{\v{s}}ek, Micha Sharir, and Emo Welzl.
\newblock A subexponential bound for linear programming.
\newblock {\em Algorithmica}, 16(4/5):498--516, 1996.

\bibitem{melekopoglou1994complexity}
Mary Melekopoglou and Anne Condon.
\newblock On the complexity of the policy improvement algorithm for {M}arkov
  decision processes.
\newblock {\em ORSA Journal on Computing}, 6(2):188--192, 1994.

\bibitem{murty1980computational}
Katta~G. Murty.
\newblock Computational complexity of parametric linear programming.
\newblock {\em Mathematical Programming}, 19(1):213--219, 1980.

\bibitem{puterman1994markov}
Martin~L. Puterman.
\newblock {\em Markov decision processes: discrete stochastic dynamic
  programming}.
\newblock John Wiley \& Sons, 1994.

\bibitem{schurr2004finding}
Ingo Schurr and Tibor Szab{\'o}.
\newblock Finding the sink takes some time: An almost quadratic lower bound for
  finding the sink of unique sink oriented cubes.
\newblock {\em Discrete \& Computational Geometry}, 31(4):627--642, 2004.

\bibitem{smale2000mathematical}
Steve Smale.
\newblock Mathematical problems for the next century.
\newblock {\em Mathematics: frontiers and perspectives}, pages 271--294, 2000.

\bibitem{spielman2004smoothed}
Daniel~A. Spielman and Shang-Hua Teng.
\newblock Smoothed analysis of algorithms: Why the simplex algorithm usually
  takes polynomial time.
\newblock {\em Journal of the ACM}, 51(3):385--463, 2004.

\bibitem{szabo2001unique}
Tibor Szab{\'o} and Emo Welzl.
\newblock Unique sink orientations of cubes.
\newblock In {\em Proceedings of the 42nd IEEE Symposium on Foundations of
  Computer Science (FOCS)}, pages 547--555, 2001.

\bibitem{todd1986polynomial}
Michael~J Todd.
\newblock Polynomial expected behavior of a pivoting algorithm for linear
  complementarity and linear programming problems.
\newblock {\em Mathematical Programming}, 35(2):173--192, 1986.

\bibitem{vershynin2009beyond}
Roman Vershynin.
\newblock Beyond hirsch conjecture: walks on random polytopes and smoothed
  complexity of the simplex method.
\newblock {\em SIAM Journal on Computing}, 39(2):646--678, 2009.

\end{thebibliography}

\end{document}